\documentclass[acmsmall,nonacm]{acmart}

\usepackage[T1]{fontenc}
\usepackage{dsfont}
\usepackage{bm}
\usepackage{graphicx}
\usepackage{pgf}
\usepackage{caption}
\usepackage{hyperref}

\AtBeginDocument{%
  \providecommand\BibTeX{{%
    \normalfont B\kern-0.5em{\scshape i\kern-0.25em b}\kern-0.8em\TeX}}}

\def\BibTeX{{\rm B\kern-.05em{\sc i\kern-.025em b}\kern-.08em
    T\kern-.1667em\lower.7ex\hbox{E}\kern-.125emX}}

\newcommand{\thmnamerestated}{placeholder}
\usepackage{natbib}
\usepackage{enumitem}
\allowdisplaybreaks

\renewcommand{\pod}{Po$d$}
\newcommand{\expect}{\mathbb{E}}
\renewcommand{\Pr}{\mathbb{P}}

\newcommand{\terml}{\mathcal{T}_1}
\newcommand{\termr}{\mathcal{T}_2}
\newcommand{\probdrop}{p_{\mathrm{d}}}
\newcommand{\probcond}{p_{\mathrm{c}}}

\begin{document}

\title[Achieving Zero Asymptotic Queueing Delay for Parallel Jobs]
{Achieving Zero Asymptotic Queueing Delay for Parallel Jobs}
\author{Wentao Weng}
\email{wwt17@mails.tsinghua.edu.cn}
\affiliation{%
    \institution{Institute for Interdisciplinary Information Sciences, Tsinghua University}
    \state{Beijing}
    \postcode{100084}
    \country{China}
}
\author{Weina Wang}
\email{weinaw@cs.cmu.edu}
\affiliation{%
    \institution{Computer Science Department, Carnegie Mellon University}
    \state{PA}
    \postcode{15213-3890}
    \country{USA}
}
\begin{abstract}
Zero queueing delay is highly desirable in large-scale computing systems. Existing work has shown that it can be asymptotically achieved by using the celebrated Power-of-$d$-choices (\pod) policy with a probe overhead $d = \omega\left(\frac{\log N}{1-\lambda}\right)$, and it is impossible when $d = O\left(\frac{1}{1-\lambda}\right)$, where $N$ is the number of servers and $\lambda$ is the load of the system. However, these results are based on the model where each job is an \emph{indivisible} unit, which does not capture the parallel structure of jobs in today's predominant parallel computing paradigm.

This paper thus considers a model where each job consists of a batch of parallel tasks. Under this model, we propose a new notion of zero (asymptotic) queueing delay that requires the job delay under a policy to approach the job delay given by the max of its tasks' service times, i.e., the job delay assuming its tasks entered service right upon arrival.  This notion quantifies the effect of queueing on a \emph{job level} for jobs consisting of multiple tasks, and thus deviates from the conventional zero queueing delay for single-task jobs in the literature.

We show that zero queueing delay for parallel jobs can be achieved using the \emph{batch-filling policy} (a variant of the celebrated \pod\ policy) with a probe overhead $d = \omega\left(\frac{1}{(1-\lambda)\log k}\right)$ in the sub-Halfin-Whitt heavy-traffic regime, where $k$ is the number of tasks in each job { and $k$ properly scales with $N$ (the number of servers)}. This result demonstrates that for \emph{parallel jobs}, zero queueing delay can be achieved with a smaller probe overhead. We also establish an impossibility result: we show that zero queueing delay cannot be achieved if $d = \exp\left({o\left(\frac{\log N}{\log k}\right)}\right)$.
Simulation results are provided to demonstrate the consistency between numerical results and theoretical results under reasonable settings, and to investigate gaps in the theoretical analysis.

\end{abstract}

\maketitle

\section{Introduction}\label{sec:intro}
In view of the rise in the amount of latency-critical workloads in today's datacenters \cite{b36,b34}, load-balancing policies with ultra-low latency have attracted great attention (see, e.g., \cite{b26,b10,b11,b5,LiuYin_20}).  In particular, it is highly desirable to have a policy under which the delay due to queueing is minimal.

In a classical setting of load-balancing, the celebrated greedy policy, Join-the-Shortest-Queue (JSQ), achieves a minimal queueing delay in the sense that the queueing delay is \emph{diminishing} as the system becomes large, even in heavy-traffic regimes \cite{b16,b25,b26}.  Therefore, we say that JSQ achieves a \emph{zero (asymptotic) queueing delay}.  Specifically, consider a system with $N$ servers where jobs arrive into the system following a Poisson process.   Each server has its own queue and serves jobs in the queue in a First-Come-First-Serve manner.  Under JSQ, each incoming job will be assigned to a server with the shortest queue length.  Then the expected time (in steady state) a job spends in the queue \emph{before} entering service goes to zero as $N$ goes to infinity.

However, a drawback of JSQ is that it has a high communication overhead, which can cancel out its advantage of achieving zero queueing delay.  For assigning each job, JSQ requires the knowledge of the queue-length information of all the $N$ servers, which will be referred to as having a \emph{probe overhead} of $N$.  In a typical cluster of servers, $N$ is in the tens of thousands range, resulting in intolerable delay due to communication \cite{b36,b34}.

A load-balancing algorithm that provides tradeoffs between queueing delay and communication overhead is the Power-of-$d$-choices (\pod) policy \cite{b18,b17}.  For each incoming job, \pod\ selects $d$ queues out of $N$ queues uniformly at random, and assigns the job to a shortest queue among the $d$ selected queues.  Therefore, \pod\ has a probe overhead of $d$.  It is easy to see that when $d=N$, \pod\ coincides with JSQ, thus achieving a zero queueing delay.  However, a fundamental question is: \emph{Can zero queueing delay be achieved by \pod\ with a $d$ value smaller than $N$?  Or, what is the smallest $d$ for achieving zero queueing delay?}

This question has been recently answered in a line of research \cite{b26,b11,b5,LiuYin_20}.  In particular, the following results are the most relevant to our paper.  Suppose the job arrival rate is $N\lambda$ and job service times are exponentially distributed with rate $1$.  Then the load of the system is $\lambda$.  Consider a heavy-traffic regime with $\lambda=1-\beta N^{-\alpha}$, where $\alpha$ and $\beta$ are constants with $0<\beta\le 1$ and $0<\alpha <1$.  It has been shown that \pod\ achieves zero queueing delay when $d=\Omega\left(\frac{\log N}{1-\lambda}\right)$ for $\alpha\in(0,0.5)$ and when $d=\Omega\left(\frac{\log^2 N}{1-\lambda}\right)$ for $\alpha\in[0.5,1)$; and it does not have zero queueing delay when $d=O\left(\frac{1}{1-\lambda}\right)$.
However, although these prior results provide great insights into achieving zero queueing delay, they are all for the classical setting where each job is an indivisible unit.

In today's applications, parallel computing is becoming increasingly popular to support the rapidly growing data volume and computation demands, especially in large scale clusters that support data-parallel frameworks such as \cite{Vavilapalli_2013,Zaharia_2010}. A job with a parallel structure is no longer a single unit, but rather has multiple components that can run in parallel.
{
In particular, the vast number of data analytic and scientific computing workloads are parallel or embarrassingly parallel \cite{Ousterhout_2013,b34, Jonas_2017}. Additional application examples include data replications in distributed file systems \cite{Lakshman_2010,Decandia2007} and hyper-parameter tuning and Monte-Carlo search in machine learning \cite{Jonas_2017,Neiswanger_2013}. 
}

In this paper, inspired by this emerging paradigm of parallel computing, we revisit the fundamental question on the minimum probe overhead needed for achieving zero queueing delay, and answer it under parallelism.
To capture the parallel structure, we consider a model where each job consists of $k$ tasks that can run on different servers in parallel.  We assume that task service times are independent and exponentially distributed with rate $1$.
{
Under such a model, we focus on delay performance on a job level, i.e., we are interested in \emph{job delay}, which is the time from when a job arrives until all of its tasks are completed.  We choose this performance metric since usually a job is a meaningful unit for users.
In fact, minimizing the delay of jobs, rather than the delay of their tasks, is the design goal of many practical schedulers \cite{Gog_2016,Boutin_2014,b34,Delimitrou_2015}.
}


{
We reiterate that we consider the asymptotic regime that $N\to\infty$.  We assume that $k$, the number of tasks per job, properly scales with $N$.}

\subsubsection*{\textbf{Zero queueing delay for parallel jobs}}
The term ``zero queueing delay'' is usually used to refer to the regime where the delay due to queueing is minimal, i.e., where jobs barely wait behind each other and are thus only subject to delay due to their inherent sizes.  In the non-parallel model, it is clear that the delay due to queueing for a job is just the time a job spends waiting in the queue.  However, when a job consists of \emph{multiple} tasks, quantifying the delay due to queueing is more complicated since different tasks experience different queueing times.

In this paper, we propose the following notion of zero queueing delay for parallel jobs.  Let $X_1,X_2,\dots,X_k$ denote the service times of a job's $k$ tasks.  Then if a job does not experience any queueing, its delay is given by $T^*=\max\{X_1,X_2,\dots,X_k\}$.  This is the job delay when all the tasks of the job enter service immediately, so we call it the \emph{inherent delay}.
{
Note that here the inherent delay is \emph{not} the total size of all the tasks of a job, but rather the delay of the job when it is parallelized.
}
Let $T$ denote the delay of a job in steady state under a load-balancing policy.  Then the delay due to queueing can be characterized by the difference $\expect[T-T^*]$.  We say zero queueing delay is achieved if
\begin{equation}\label{eq:zero-queue-delay}
\frac{\expect[T-T^*]}{\expect[T^*]}\rightarrow 0\quad \text{as }N\to\infty,
\end{equation}
i.e., the queueing delay takes a diminishing fraction of the inherent delay.

Our notion of zero queueing delay recovers the conventional notion for non-parallel jobs when $k=1$.  However, it is different from the requirement that under the parallel job model, all the tasks of a job should have zero queueing delay.  Such a requirement is rather strong since all the tasks would need to be assigned to empty queues \emph{simultaneously}.  We will discuss this alternative notion in more detail in Section~\ref{sec:alter}.

\subsubsection*{\textbf{Probe overhead and batch-filling policy}}
When a job arrives into the system, a task-assigning policy samples some queues to obtain their queue length information, and then decides how to assign the $k$ tasks to the sampled servers.  If the policy samples $kd$ queues, then we say its \emph{probe overhead} \cite{b3,b34} is $d$ since $d$ is the average number of samples per task.

In this paper, we focus on a policy called \emph{batch-filling}.  It samples $kd$ queues for an incoming job and then assigns its tasks one by one to the shortest queue, where the queue length is updated after every task assignment.  Batch-filling has been shown to outperform the per-task version of \pod\ and also a policy called batch-sampling \cite{b3,b34}.

{
Note that the queueing dynamics under batch-filling with a probe overhead of $d$ is also very different from that under the policy that runs Po-$kd$ for each task, although in both polices a task gets to join the shortest queue among a set of $kd$ queues.  For this per-task version of Po-$kd$, tasks of the same job pick their own $kd$ queues independently.  Then it could happen that some task picks $kd$ lightly loaded servers while some other task lands in $kd$ highly busy servers.  While under batch-filling, all the tasks in a job experience the same set of $kd$ servers.  Therefore, the analyses of batch-filling and per-task Po-$kd$ will be very different.
}

\subsubsection*{\textbf{Challenges and our results}}

The parallel structure of jobs makes a load-balancing system more challenging to analyze in the following two aspects: 
(i) The delay of an incoming job in steady state (tagged job) depends on the system state (queue lengths) in a more intricate way since its tasks may be assigned to different queues.  (ii) The dynamics of the system state is complicated by the simultaneous arrival of a batch of tasks and the coordination in assigning tasks.
{
Due to these intricacies, existing techniques for analyzing non-parallel models do not directly carry over to parallel models.
}

We address these difficulties by first deriving a sufficient condition on the state for an incoming job to achieve zero queueing delay.  Notably, this condition involves all the servers whose queue lengths range from zero to a threshold that is in the order of $o(\log k)$.  This is in contrast to the condition for the non-parallel model, which only depends on the fraction of idle servers.  Based on this first step, we recognize that we only need to understand the system dynamics in terms of whether the steady state concentrates around the set of desirable states that satisfy the sufficient condition.  Towards this end, a key in our analysis is an interesting state-space collapse result we discover, which enables us to use the powerful framework of Stein's method \cite{b2,b14}.


Specifically, we consider a system with a job arrival rate of $N\lambda/k$.  We focus on a heavy-traffic regime where the load $\lambda=1-\beta N^{-\alpha}$ with $0<\beta \le 1$ and $0<\alpha<0.5$, i.e., the sub-Halfin-Whitt regime.  Note that the larger $\alpha$ is, the faster the load approaches $1$ as $N\to\infty$.  All the order notation and asymptotic results in this paper are with respect to the regime that $N\to\infty$.

Our main result is that zero queueing delay is \emph{achievable} when the probe overhead $d$ satisfies
\begin{equation}\label{eq:d-min}
d=\omega\left(\frac{1}{(1-\lambda)\log k}\right),
\end{equation}
where the number of tasks $k$ satisfies $k=o\left(\frac{N^{0.5-\alpha}}{\log^2 N}\right)$ and $\frac{k}{\log k} = \Omega(\log N)$.  For example, this includes $k=\log^2 N$ and $k=N^{0.1}$ with $\alpha<0.4$.
Recall that for the \emph{non-parallel} model, a lower bound result is that zero queueing \emph{cannot} be achieved when the probe overhead is $O\left(\frac{1}{1-\lambda}\right)$.  In contrast, we can see that for \emph{parallel} jobs, the probe overhead in \eqref{eq:d-min} can be orderly smaller than $\frac{1}{1-\lambda}$.

We comment that this reduction in probe overhead reflects the overall effect of parallelization on the system.
There are several factors at play that are brought by parallelization all together, making it hard to quantify their individual effects.  First, for tasks of the same job, the probe overhead quota is pooled together and their assignment is coordinated, leading to a more effective use of the state information.  Second, a job with parallel tasks can better tolerate task delays since the job delay is anyway determined by the slowest task.  Furthermore, work arrives to the system in a more bursty fashion under parallelization due to the batch effect.
{
Such an effect of parallelization has also been investigated in some recent papers \cite{b38,ShnSto_20}.  But generally, understanding parallel jobs is a much underexplored research area.
}

To complement our achievability results, we also prove an impossibility result on the minimum probe overhead needed: zero queueing delay can not be achieved if
\begin{equation}
d=e^{o\left(\frac{\log N}{\log k}\right)},
\end{equation}
where $k$ satisfies that $k = e^{o\left(\sqrt{\log N}\right)}$ and $k = \omega(1)$.  To establish this lower bound, we utilize the tail bound given by a Lyapunov function in a ``reversed'' way.

To the best of our knowledge, our paper is the first one that characterizes zero queueing delay on a job level for jobs with parallel tasks.  The very limited amount of prior work that does study parallel jobs only has fluid-level optimality and only considers a constant load. Furthermore, we develop a new technique for lower-bounding queues, which may be of separate interest itself given the scarcity of lower-bounding techniques in queueing systems in general.

\subsubsection*{\textbf{A reminder of Bachmann--Landau asymptotic notation}}
Since Bachmann--Landau asymptotic notation is heavily used in this paper, here we briefly recap the definitions for ease of reference.  For two real-valued functions $f$ and $g$ of $N$ where $g$ takes positive values, we say that $f(N) = O(g(N))$ if there exists a positive number $M$ such that $|f(N)|\le M\cdot g(N)$ for large enough $N$, or equivalently if $\lim \sup_{N \to \infty}\left|\frac{f(N)}{g(N)}\right| < \infty$.  We say that $f(N) = o(g(N))$ if $\lim_{N \to \infty} \frac{f(N)}{g(N)} = 0$; $f(N) = \Omega(g(N))$ if $\lim \inf_{N \to \infty} \frac{f(N)}{g(N)} > 0$; and $f(N) = \omega(g(N))$ if $\lim \inf_{N \to \infty} \left|\frac{f(N)}{g(N)}\right| = \infty$. In this paper, the asymptotic regime is when $N$, the number of servers, goes to infinity.

\subsubsection*{\textbf{Related work}}
Load-balancing systems for \emph{non-parallel} jobs have been extensively studied in the literature.  It is well-known that JSQ is delay-optimal under a wide range of assumptions \cite{b16,b25}.  Although getting exact-form stationary distributions is typically not feasible for most load-balancing policies, many results and approximations are known for various asymptotic regimes.  

For JSQ in heavy-traffic regimes, \citet{b12} obtain a diffusion approximation in the Halfin-Whitt regime ($\alpha = 0.5$), which has a zero queueing delay in the diffusion limit.  The convergence result in \cite{b12} is on the process level.  \citet{b24} later establish steady-state results and their results imply the convergence of the stationary distributions to the diffusion limit.  JSQ has also been studied in the nondegenerate slowdown (NDS) regime ($\alpha=1$) \cite{b37}.

The problem of achieving zero queueing delay with \pod\ has been studied in \cite{b26,b11,b5,LiuYin_20}.  \citet{b26} show through stochastic coupling that the diffusion limit of \pod\ with $d=\omega(N^{0.5}\log N)$ converges to that of JSQ in the Halfin-Whitt regime, thus resulting in a zero queueing delay.  The convergence to the diffusion limit in \cite{b26} is on the process level.   Zero queueing delay for \pod\ in steady state is first studied by \citet{b11} for the regime where $\alpha < \frac{1}{6}$, where they show that the waiting probability goes to $0$ as $N\to\infty$ when $d=\omega\left(\frac{1}{1-\lambda}\right)$.  The results are later extended to the sub-Halfin-Whitt regime $(0<\alpha<0.5)$ for both exponential and Coxian-2 service times \cite{b5,LiuYin_20} and beyond-Halfin-Whitt regime $(0.5\le\alpha<1)$ \cite{b5}, where it is shown that zero queueing delay is achieved 
when $d=\Omega\left(\frac{\log N}{1-\lambda}\right)$ for $\alpha \in (0, 0.5)$, and when $d = \Omega\left(\frac{\log^2 N}{1-\lambda}\right)$ for $\alpha \in [0.5,1)$.  The paper \cite{b11} also provides a lower bound result: the waiting probability is bounded away from $0$ when $d=O\left(\frac{1}{1-\lambda}\right)$ for $0\le\alpha<1$.

\pod\ has also been analyzed in the regime with a constant load ($\alpha=0$) as $N\to\infty$.  Mean-field analysis has been derived for a constant $d$ in \cite{b17,b18}, and \citet{b26} show $d=\omega(1)$ leads to zero queueing delay.  We remark that mean-field analysis results are also available for other policies such as Join-the-Idle-Queue (JIQ) \cite{b20,b35}, and also for delay-resource tradeoffs \cite{b10}.

To the best of our knowledge, very limited work has been done on achieving zero queueing delay for \emph{parallel jobs}, or on analyzing delay for parallel jobs in general.  Only the regime with a constant load as $N\to\infty$ has been studied.  \citet{b26} briefly touch upon this topic and show that fluid-level optimality can be achieved with probe overhead $d \geq \frac{1}{1 - \lambda - \epsilon}$ under the so-called batch-sampling policy \cite{b34}.  \citet{b3} provide limiting distributions for the stationary distributions under (batch-version) \pod, batching-sampling, and batch-filling, but have not analyzed delay of jobs.  \citet{b38} analyze job delay under a (batch-version) random-routing policy, which does not achieve zero queueing delay.  There have been no results for heavy-traffic regimes.

Finally, the techniques we use in this paper are based on Stein's method and drift-based state-space collapse.  Proposed in \cite{b32}, Stein's method has been an effective tool for bounding the distance between two distributions.  The seminal papers \cite{b2,b14,b33} build an analytical framework for Stein's method in queueing theory that consists of generator approximation, gradient bounds, and possibly state-space collapse.  The papers \cite{b2,b14} use Stein's method to study steady-state diffusion approximation, and \cite{b11,b15,b23,b24,b29,b30,b31,LiuYin_20} use Stein's method to obtain convergence rates to the mean-field limit. A similar approach has also been developed by \citet{b39}.

\begin{figure}
\begin{minipage}[b]{.45\textwidth}
\centering
    \centering
    \includegraphics[scale=0.25]{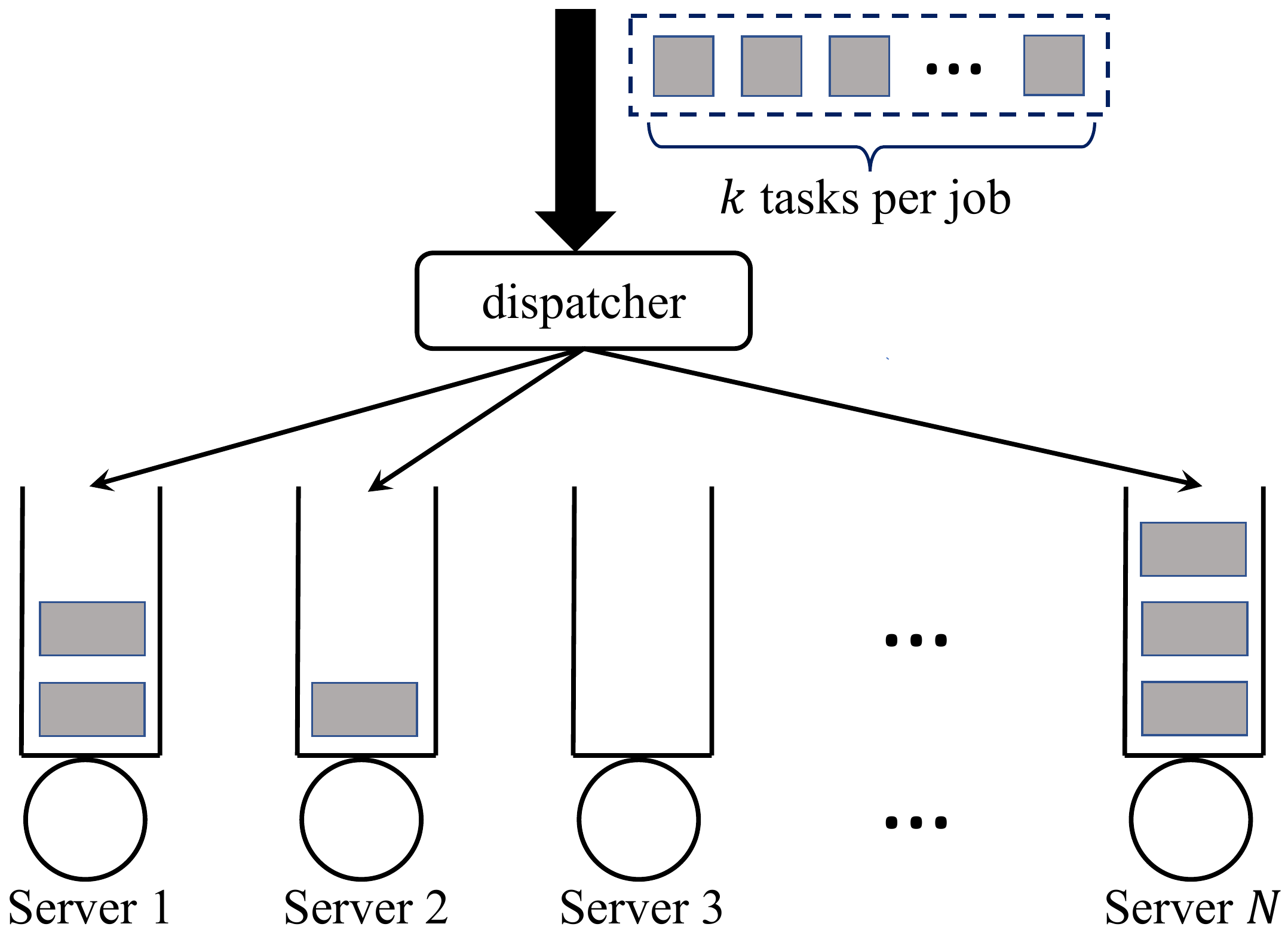}
    \caption{A $N$-server system with batch arrivals.}
    \label{fig:multiserver_model}
\end{minipage}
\hspace{0.3in}
\begin{minipage}[b]{.45\textwidth}
\centering
    \centering
    \includegraphics[scale=0.25]{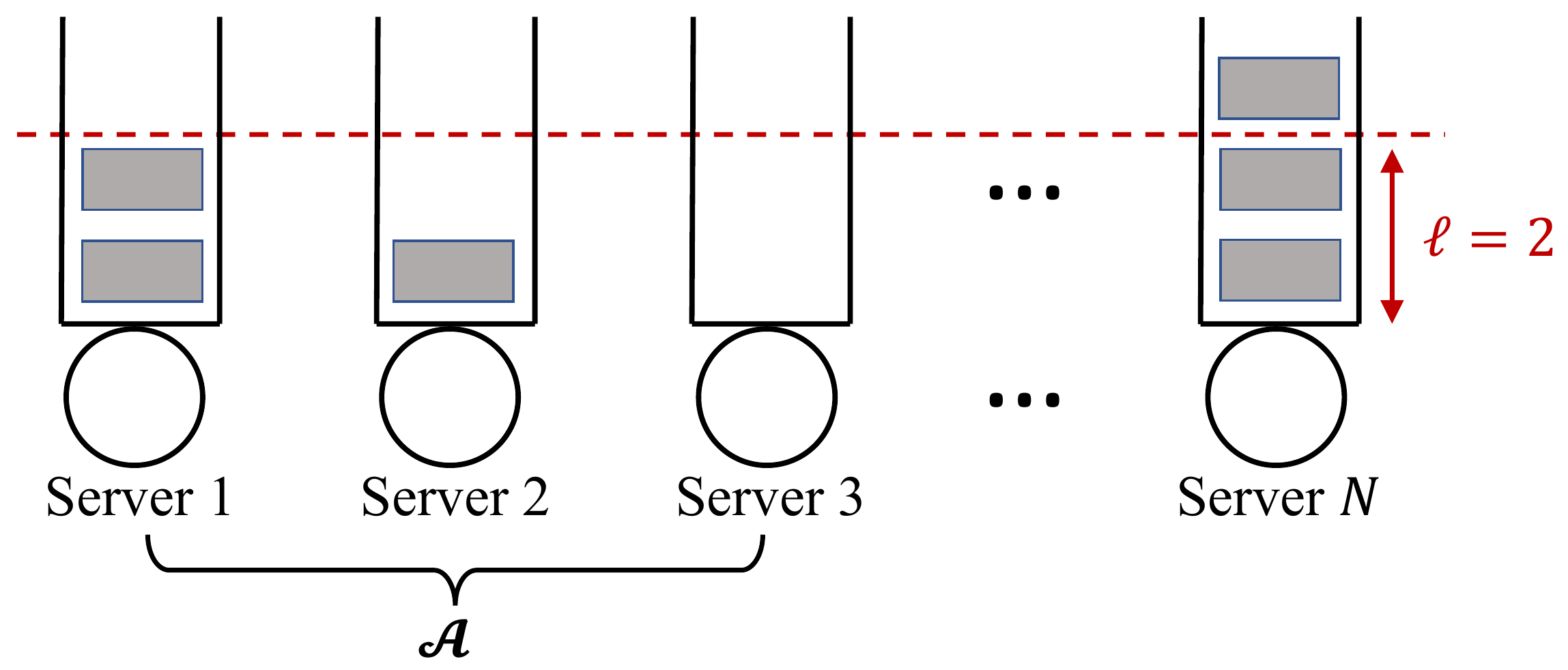}
    \caption{An example of the number of spaces below a threshold $\ell$ in a set of queues: $\ell = 2$, set of queues $\mathcal{A}=\{1,2,3\}$, and $N_{\ell}(\mathcal{A}) = 3$.}
    \label{fig:wwt-threshold}
\end{minipage}
\end{figure}

\section{Model}
We consider a system with $N$ identical servers, illustrated in Figure~\ref{fig:multiserver_model}.  Each server has its own queue and serves tasks in its queue in a First-Come-First-Serve manner.  Since each queue is associated with a server, we will refer to queues and servers interchangeably.
Jobs arrive into the system following a Poisson process.  To capture the parallel structure of jobs, we assume that each job consists of $k$ tasks that can run on different servers in parallel.
{ A job is completed when all of its tasks are completed.}  We study the large-system regime where the number of servers, $N$, becomes large, and we will let $k$ increase to infinity with $N$ to capture the trend of growing job sizes.

We denote the job arrival rate by $N\lambda/k$ and assume that the service times of tasks are independent and exponentially distributed with rate $1$.  
Then $\lambda$ is the load of the system.  We consider a heavy-traffic regime where $\lambda=1-\beta N^{-\alpha}$ with $0<\beta \le 1$ and $0<\alpha<0.5$, i.e., the so-called sub-Halfin-Whitt regime \cite{b13,LiuYin_20}.

When a job arrives into the system, we sample $kd$ queues and obtain their queue length information.  Since the average overhead is $d$ samples per task, the probe overhead is $d$.  We then assign the $k$ tasks of the job to the $kd$ selected queues using the batch-filling policy proposed in \cite{b3}.  Batch-filling assigns the $k$ tasks one by one to the shortest queue, where the queue length is updated after each task assignment.  Specifically, the task assignment process runs in $k$ rounds.  For each round, we put a task into the shortest queue among sampled queues.  We then update the queue length, and continue to the next round.

Now we give an equivalent description of batch-filling, which is useful in our analysis.  For each queue and a positive integer $\ell$, we use the \emph{number of spaces below threshold $\ell$} to refer to the quantity $\max\{\ell-\text{queue length},0\}$, i.e., the number of tasks we can put in the queue such that the queue length after receiving the tasks is no larger than $\ell$.  For a set of queues $\mathcal{A}$, we use $N_{\ell}(\mathcal{A})$ (or just $N_{\ell}$ when it is clear from the context) to denote the total number of spaces below $\ell$ in $\mathcal{A}$.  Figure~\ref{fig:wwt-threshold} gives an example of $N_{\ell}(\mathcal{A})$.
We say a task is at a \emph{queueing position} $p$ if there are $p-1$ tasks ahead of it in the queue.
With the above terminology, the batch-filling policy can be described in the following way: it finds a minimum threshold $\ell$ such that the total number of spaces below $\ell$ in the sampled queues is at least $k$.  Then it fills the $k$ tasks into these spaces from low positions to high positions.

Recall that we propose the following notion of zero queueing for parallel jobs. Let $X_1, X_2, \cdots, X_k$ be the service times of the tasks of a job. If a job does not experience any queueing, its delay is given by $T^* = \max\left\{X_1,\cdots,X_k\right\}$, which we call the \emph{inherent delay} of this job.
Then if the actual delay of the job is very \emph{close} to its inherent delay, it is as if the job almost experiences no queueing.
{We say zero queueing delay is achieved if the steady-state job delay, $T$, is larger than $T^*$ only by a diminishing fraction; i.e., if $T$ satisfies $\mathbb{E}[T-T^*]/\mathbb{E}[T^*]\to 0$ as $N\to \infty$ as in \eqref{eq:zero-queue-delay}.}
We note that as the service time of each task is exponentially distributed with mean $1$, it holds that
$$
\mathbb{E}[T^*] = H_k = \ln k + o(\ln k),
$$
where $H_k=1+\frac{1}{2}+\dots+\frac{1}{k}$ is the $k$-th harmonic number.

We make the following interesting observation, which provides a basis for our delay analysis of parallel jobs: a job can have zero queueing delay even when its tasks are assigned to non-idle servers.  In fact, we establish a necessary and sufficient condition: a job has zero queueing delay if and only if all of its tasks are at queueing positions below a threshold $h$ with $h=o(\log k)$ after assigned to servers, noting that the inherent delay is $\ln k + o(\ln k)$.  The formal proof is based on Lemma~\ref{lem:max-upper}.  This phenomenon allows us to have a zero queueing delay with low probe overhead.  But it also makes the analysis hard since it implies that there are many situations that can lead to zero queueing delay.

We assume that every queue has a finite buffer size of $b$ including the task in service. If the dispatcher routes a task to a queue with length equal to $b$, we simply discard this task and all the other tasks of the same job.  In this case, we say the job is \emph{dropped}; otherwise, we say the job is \emph{admitted}.
We remark that this assumption is not restrictive for the following two reasons:  (1) our results hold for a very large range of $b$ (see Theorem~\ref{main_theorem_1}); and (2) the probability of discarding a job is very small (see Theorem~\ref{theorem zero delay}).

To represent the state of the system, let $S_i(t)$ denote the fraction of servers that have at least $i$ jobs at time $t$, where $0\le i\le b$.  Note that it always holds $S_0(t)=1$.  Then $\bm{S}(t) = (S_0(t),S_1(t),\cdots,S_b(t))$ forms a continuous-time Markov chain (CTMC) since batch-filling is oblivious to labels of servers.  The state space is as follows:
\begin{align*}
\mathcal{S} &= \left\{\bm{s}=(s_0,s_1,s_2,\cdots,s_b) : 1 = s_0 \geq s_1 \geq s_2 \geq \cdots s_b,
\text{where } Ns_i \in \mathbb{N}, \forall 1 \leq i \leq b\right\}.
\end{align*}
It can be verified that $\left\{\bm{S}(t) \colon t \geq 0\right\}$ is irreducible and positive recurrent, thus having a unique stationary distribution.  Let $\pi_{\bm{S}}$ denote this stationary distribution, and let $\bm{S} = (S_1,\cdots,S_b)$ be a random element with distribution $\pi_{\bm{S}}$.

\section{Main Results}
Our main results provide bounds on queue lengths and delay, which lead to corresponding conditions on the probe overhead for achieving zero queueing delay.  We divide our results into \emph{achievability} and \emph{impossibility} results.  Again, all the asymptotics are with respect to the regime that the number of servers, $N$, goes to infinity.

\subsubsection*{\textbf{Achievability Results}}
In Theorem~\ref{main_theorem_1}, we give an upper bound that characterizes $\expect\left[\sum_{i=1}^b S_i\right]$, which is equal to the average expected number of tasks per server.  This upper bound underpins our analysis of job delay.

\begin{theorem} \label{main_theorem_1}
Consider a system with $N$ servers where each job consists of $k$ tasks.  Let the load be $\lambda=1-\beta N^{-\alpha}$ with $0<\beta \le 1$ and $0<\alpha <0.5$.  Under the batch-filling policy with a probe overhead of $d$ such that $d \geq \frac{8}{(1-\lambda)h}$ for some $h = o(\log k)$ and $h = \omega(1)$, it holds that
\begin{equation}
\mathbb{E}\left[\max\left\{\sum_{i=1}^b S_i - h\left(1 - \frac{1}{2}\beta N^{-\alpha}\right),0\right\}\right] \leq \frac{5}{\sqrt{N} \log N},
\end{equation}
where we assume that $k$ satisfies $k = o\left(\frac{N^{0.5 - \alpha}}{\log^2 N}\right)$ and $\frac{k}{\log k} = \Omega(\log N)$, the buffer size $b$ is given by $b = \min\left\{N^{\alpha},\frac{N^{0.5-\alpha}}{k}\right\}$, and $N$ is sufficiently large.
\end{theorem}

We remark that the $h = o(\log k)$ in this theorem represents the threshold position we pointed out for zero queueing delay, i.e., a job has zero queueing delay if all of its tasks are at queueing positions below $h$ after assigned to servers.

The upper bound in Theorem~\ref{main_theorem_1} enables us to analyze the probability that all the tasks of an incoming job end up in positions below $h$ under batch-filling, which further leads to the zero queueing delay result below in Theorem~\ref{theorem zero delay}.  Recall that the buffer size $b$ of each queue is finite, so a job will get dropped if at least one of its tasks is assigned to a queue with a full buffer.  We denote the probability of dropping an incoming job in steady state by $\probdrop$.

\begin{theorem}\label{theorem zero delay}
Under the assumptions of Theorem \ref{main_theorem_1}, the steady-state delay of jobs that are admitted under batch-filling satisfies that
\begin{equation*}
\expect[T\;|\;\textup{admitted}]=\ln k + o(\ln k),
\end{equation*}
with a dropping probability $\probdrop \le \frac{11}{b\sqrt{N}\log N}$ when $N$ is sufficiently large.
\end{theorem}
{Theorem~\ref{theorem zero delay} thus implies that zero queueing delay for parallel jobs can be achieved} with a probe overhead $d = \omega\left(\frac{1}{(1-\lambda)\log k}\right)$.  This breaks the lower bound of $\omega\left(\frac{1}{1-\lambda}\right)$ for achieving zero queueing delay for non-parallel jobs, i.e., single-task jobs \cite{b11}, as we discussed in Section~\ref{sec:intro}.

\subsubsection*{\textbf{Impossibility Results}}
To complement the achievability results, below we investigate when zero queueing delay cannot be achieved.  In Theorem~\ref{main_theorem_2}, we find conditions under which $\sum_{i=1}^h S_i$ is lower bounded with a constant probability.

\begin{theorem} \label{main_theorem_2}
Consider a system with $N$ servers where each job consists of $k$ tasks.  Let the load be $\lambda=1-\beta N^{-\alpha}$ with $0<\beta \le 1$ and $0<\alpha <0.5$.  Assume that buffers have unlimited sizes and $k$ satisfies that $k = e^{o\left(\sqrt{\log N}\right)}$ and $k = \omega(1)$.  Under the batch-filling policy with a probe overhead $d$ such that $d = e^{o\left(\frac{\log N}{\log k}\right)}$ and for any $h$ with $h=O(\log k)$, it holds that when $N$ is sufficiently large,
\begin{equation}
\Pr\left\{\sum_{i=1}^h S_i \geq h - \frac{1}{3d}\right\} \geq \frac{1}{4e^2}.
\end{equation}
\end{theorem}

The lower bound on $\sum_{i=1}^h S_i$ in Theorem~\ref{main_theorem_2} guarantees that an incoming job will have a significant delay in addition to its inherent delay, and thus fails to have zero queueing delay.  This result is formally stated in Theorem~\ref{theorem lower d} below.

\begin{theorem}\label{theorem lower d}
Under the assumptions of Theorem~\ref{main_theorem_2}, the steady-state job delay, $T$, satisfies that
\begin{equation}
\mathbb{E}[T] \geq 2 \ln k
\end{equation}
when $N$ is sufficiently large.  Therefore, to achieve zero queueing delay, the probe overhead $d$ needs to be at least $e^{\Omega\left(\frac{\log N}{\log k}\right)}$.
\end{theorem}

\section{Proofs for Achievability Results (Theorems~\ref{main_theorem_1} and \ref{theorem zero delay})}
{
Before we dive into the proofs of Theorems~\ref{main_theorem_1} and \ref{theorem zero delay}, we first develop more understanding of zero queueing delay on a job level through Lemmas~\ref{lem:max-upper} and \ref{lem:filling}.
Due to the space limit, the proofs of the lemmas are presented in Appendix~\ref{app:lemmas-achievability}.
Then we provide a proof sketch for Theorems~\ref{main_theorem_1} and \ref{theorem zero delay} to outline the main steps.
Detailed proofs of Theorems~\ref{main_theorem_1} and \ref{theorem zero delay} are presented in Sections~\ref{sec:proof-main_theorem_1} and \ref{sec:proof-theorem zero delay}, respectively.
Throughout this section, we assume that the assumptions in Theorem~\ref{main_theorem_1} hold.

\subsubsection*{\textbf{Zero queueing delay and queue lengths}}
Lemma~\ref{lem:max-upper} below gives an upper bound on the expected job delay given the \emph{lengths of the queues} that the tasks of a job are assigned to.
Specifically, suppose the $k$ tasks of a job are sent to $m$ queues ($m\le k$) with queue lengths $n_1,n_2,\dots,n_m$, where the queue lengths have included these newly arrived tasks.  Note that multiple tasks of the job could be sent to the same queue, but to compute the job delay, we only need to consider the last task of the job in that queue.  Let $Y_i$ with $1\le i\le m$ denote the delay of the last task of the job in queue $i$.  Then the job delay can be written as $\max\left\{Y_1,\cdots,Y_m\right\}$.  Lemma~\ref{lem:max-upper} gives an upper bound on $\mathbb{E}[\max\left\{Y_1,\cdots,Y_m\right\}]$.
\begin{lemma}\label{lem:max-upper}
Consider $m$ independent random variables $Y_1,\cdots,Y_m$ with $m\le k$, where each $Y_i$ ($1\le i\le m)$ is the sum of $n_i$ i.i.d.\ random variables that follow the exponential distribution with rate $1$.  In the asymptotic regime that $k$ goes to infinity, if $\max\left\{n_1,\cdots,n_m\right\}=o(\log k)$, then
$$
\mathbb{E}[\max\left\{Y_1,\cdots,Y_m\right\}] \leq \ln k + o(\ln k).
$$
\end{lemma}

The upper bound in Lemma~\ref{lem:max-upper} implies that a sufficient condition for this job to have zero queueing delay is that the lengths of the queues that its tasks are assigned to are of order $o(\log k)$.  As we pointed out earlier, this is different from the single-task job model since here zero queueing delay on a job level allows non-zero queueing delay for each of the tasks.

\subsubsection*{\textbf{Zero queueing delay and states}}
Lemma~\ref{lem:filling} below establishes a condition on the \emph{state seen by a job arrival} for all of its tasks to be assigned to queues of length $o(\log k)$ with high probability, which is a sufficient condition for the job to have zero queueing delay by Lemma~\ref{lem:max-upper}. Specifically, we consider the event that all the $k$ tasks of an incoming job are assigned to queueing positions below some threshold value $\ell$, and let this event be denoted by $\mathrm{FILL}_{\ell}$.  Lemma~\ref{lem:filling} shows that $\mathrm{FILL}_{\ell}$ happens with high probability given a proper condition on the state $\bm{s}$ for several values of interest for $\ell$.  Note that if we take $\ell = h$, which is $o(\log k)$, then $\mathrm{FILL}_{\ell}$ leads to zero queueing delay.  But Lemma~\ref{lem:filling} is more general in the sense that it allows other values for $\ell$, which is essential for other parts of the proofs including proving a state-space collapse result (Lemma~\ref{lem:ssc}) and bounding the dropping probability (Theorem~\ref{theorem zero delay}).
\begin{lemma}[Filling Probability]\label{lem:filling}
Under the assumptions of Theorem~\ref{main_theorem_1}, given that the system is in a state $\bm{s}$ such that
\begin{equation}\label{eq:cond-filling}
\sum_{i=1}^{\ell} s_i \leq \ell\left(1 - \frac{1}{4}\beta N^{-\alpha}\right),
\end{equation}
the probability of the event $\mathrm{FILL}_{\ell}$ for any $\ell\in\{h-1,h,b\}$ can be bounded as $\Pr\left\{\mathrm{FILL}_{\ell}\right\} \geq 1 - \frac{1}{N}$ when $N$ is sufficiently large.
\end{lemma}

Here we provide an intuitive explanation for the condition \eqref{eq:cond-filling} when $\ell = h$. When a job arrives and sees state $\bm{s}$, if we choose one queue uniformly at random from all the queues, then the probability for the chosen queue to have a length of $i$ is $s_i-s_{i+1}$.  So the expected number of spaces below position $h$ in the chosen queue is $\sum_{i=0}^{h} (h-i)(s_i-s_{i+1})=h-\sum_{i=1}^{h} s_i$.  
The batch-filling policy samples $kd$ queues.  Thus the total expected number of spaces below position $h$ in the $kd$ sampled queues is $kd\left(h-\sum_{i=1}^{h} s_i\right)$.  To fit all the $k$ tasks of the incoming job to positions below $h$, we need $k\le kd\left(h-\sum_{i=1}^h s_i\right)$, which becomes the following condition when $d \geq\frac{8}{(1-\lambda)h}=\frac{8N^{\alpha}}{\beta h}$ as required in Theorem~\ref{main_theorem_1}:
\begin{equation*}
\sum_{i=1}^h s_i \leq h\left(1 - \frac{1}{8}\beta N^{-\alpha}\right).
\end{equation*}
We strengthen this requirement to the condition $\sum_{i=1}^h s_i \leq h\left(1 - \frac{1}{4}\beta N^{-\alpha}\right)$ to obtain a high-probability guarantee using concentration bounds based on Hoeffding's results on sampling without replacement \cite[Theorem~4]{Hoe_63}.

\subsubsection*{\textbf{Proof sketch for Theorems~\ref{main_theorem_1} and \ref{theorem zero delay}}}
We start by setting the goal to be proving the zero queueing delay result in Theorem~\ref{theorem zero delay}, and we will see how Theorem~\ref{main_theorem_1} emerges as an essential characterization of the system that is needed for Theorem~\ref{theorem zero delay}.

Considering the condition in Lemma~\ref{lem:filling} on the system state, we upper bound the steady-state job delay $T$ in the following way:
\begin{align}
\mathbb{E}[T] &\leq \mathbb{E}\left[T \;\middle|\; \sum_{i=1}^h S_i \leq h\left(1 - \frac{1}{4}\beta N^{-\alpha}\right)\right]\label{eq:T-upper-1}\\ 
&\mspace{21mu}+ \mathbb{E}\left[T \;\middle|\; \sum_{i=1}^h S_i > h\left(1 - \frac{1}{4}\beta N^{-\alpha}\right)\right]\cdot\Pr\left\{\sum_{i=1}^h S_i > h\left(1 - \frac{1}{4}\beta N^{-\alpha}\right)\right\},\label{eq:T-upper-2}
\end{align}
where we have used the fact that $\Pr\mspace{-3mu}\left\{\mspace{-3mu}\sum_{i=1}^h S_i \le h\left(1 - \frac{1}{4}\beta N^{-\alpha}\right)\right\}\le 1$.  We can easily bound the first summand \eqref{eq:T-upper-1} using Lemma~\ref{lem:filling} since this is the case where all the tasks of an incoming jobs are sent to queues with lengths no larger than $h$, which satisfies $h=o(\log k)$ and thus results in zero queueing delay.

We now focus on bounding the second summand \eqref{eq:T-upper-2}, for which it suffices to show that the probability $\Pr\left\{\sum_{i=1}^h S_i > h\left(1 - \frac{1}{4}\beta N^{-\alpha}\right)\right\}$ is small enough.  By the Markov inequality,
\begin{align*}
\Pr\left\{\sum_{i=1}^h S_i > h\left(1 - \frac{1}{4}\beta N^{-\alpha}\right)\right\}
\le\frac{\mathbb{E}\left[\max\left\{\sum_{i=1}^b S_i - h\left(1 - \frac{1}{2}\beta N^{-\alpha}\right), 0\right\}\right]}{\frac{1}{4}\beta N^{-\alpha}}.
\end{align*}
It then boils down to bounding $\mathbb{E}\left[\max\left\{\sum_{i=1}^b S_i - h\left(1 - \frac{1}{2}\beta N^{-\alpha}\right), 0\right\}\right]$, which is what Theorem~\ref{main_theorem_1} achieves.

To prove Theorem~\ref{main_theorem_1}, we follow the general framework of Stein's method (see, e.g., \cite{b14,LiuYin_20}).  The main idea is to couple our Markov chain $\left\{\bm{S}(t) \colon t \geq 0\right\}$ with an auxiliary process that is easier to analyze, and bound their difference through generator approximation.  In particular, we compare the dynamics of $\sum_{i=1}^b S_i(t)$ with a continuous function $x(t)$ given by the following simple fluid model as our auxiliary process:
\begin{equation*}
\dot{x}(t) =(-\delta)\mathds{1}_{\{x>0\}},
\end{equation*}
where $\delta$ is a properly chosen parameter that reflects the drift of $\sum_{i=1}^b S_i(t)$.
We reiterate that a key in our analysis is a novel state-space collapse result (Lemma~\ref{lem:ssc}) that we establish, which characterizes how balanced the queues are from a job's point of view.

Combining the arguments above for bounding \eqref{eq:T-upper-1} and \eqref{eq:T-upper-2}, we can conclude that the steady-state job delay $\expect[T]$ achieves zero queueing delay.
}

\subsection{Proof of Theorem \ref{main_theorem_1}}\label{sec:proof-main_theorem_1}
\begin{proof}
As explained in the proof sketch, we compare our system with the following fluid model:
\begin{equation}\label{eq:fluid}
\dot{x}(t) =(-\delta)\mathds{1}_{\{x>0\}},
\end{equation}
where $x(t)$ is continuous and $\delta=\frac{(k+1)\log N}{\sqrt{N}}$.  
{When viewed as a continuous-time Markov chain, this fluid model (with a possibly random initial state) can be described by its generator \cite{EthKur_86}, denoted as $\overline{G}$ and given by}
\begin{equation*}
\overline{G}g(x)=g'(x)\cdot(-\delta)\mathds{1}_{\{x>0\}}
\end{equation*}
for any differentiable function $g$. Recall that we will compare the dynamics of $\sum_{i=1}^b S_i(t)$ in our load-balancing system with $x(t)$.

\begin{sloppypar}
The quantity of interest in Theorem~\ref{main_theorem_1} is $\expect\left[\max\left\{\sum_{i=1}^b S_i - \eta, 0\right\}\right]$, where we have used the notation $\eta=h\left(1 - \frac{1}{2}\beta N^{-\alpha}\right)$ for conciseness.  Recall that $\bm{S}$ follows the stationary distribution of $\left\{\bm{S}(t) \colon t \geq 0\right\}$.  To couple $\left\{\bm{S}(t) \colon t \geq 0\right\}$ with the fluid model, we solve for a function $g$ such that
\begin{equation}\label{eq:Poisson-equation}
\begin{split}
\overline{G}g(x)&=\max\left\{x-\eta, 0\right\},\\
g(0)&=0.
\end{split}
\end{equation}
It is not hard to see that the solution is
\begin{equation}
g(x) = \frac{(x - \eta)^2}{2(-\delta)}\mathds{1}_{\left\{x \geq \eta\right\}}.
\end{equation}
\end{sloppypar}

\begin{sloppypar}
Now we utilize this function $g$ to bound $\expect\left[\max\left\{\sum_{i=1}^b S_i -\eta, 0\right\}\right]$ through generator approximation.  Let $G$ be the generator of $\left\{\bm{S}(t)\colon t\ge 0\right\}$.  Then
\begin{equation*}
Gg\left(\sum_{i=1}^b s_i\right) = \sum_{\bm{s}' \in \mathcal{S}} r_{\bm{s} \to \bm{s}'}\left(g\left(\sum_{i=1}^b s'_i\right) - g\left(\sum_{i=1}^b s_i\right)\right),
\end{equation*}
where $r_{\bm{s} \to \bm{s}'}$ is the transition rate from state $\bm{s}$ to $\bm{s}'$. Since $g\left(\sum_{i=1}^b s_i\right)$ is bounded on $\mathcal{S}$, it holds that 
\begin{equation}
\mathbb{E}\left[Gg\left(\sum_{i=1}^b S_i\right)\right] = 0.
\end{equation}
Combining this with the equations in \eqref{eq:Poisson-equation} gives,
\begin{align}
\mathbb{E}\left[\max\left\{\sum_{i=1}^b S_i - \eta,0\right\}\right]
&= \mathbb{E}\left[\overline{G}g\left(\sum_{i=1}^b S_i\right)\right]\nonumber\\
&= \mathbb{E}\left[\overline{G}g\left(\sum_{i=1}^b S_i\right)- Gg\left(\sum_{i=1}^b S_i\right)\right]\nonumber\\
&= \mathbb{E}\left[g'\left(\sum_{i=1}^b S_i\right)(-\delta) - Gg\left(\sum_{i=1}^b S_i\right)\right].\label{eq:generator-approx}
\end{align}
This is what is referred to as a \emph{generator approximation} since we are approximating the generator $G$ with $\overline{G}$.
\end{sloppypar}

Next we take a closer look at the term $Gg\left(\sum_{i=1}^b S_i\right)$ and derive an upper bound for \eqref{eq:generator-approx}.
Let $P_A(\bm{s})$ be the probability that a job arrival is admitted into the system given that the system is at state $s$, i.e., the probability that all the tasks of the job are routed to positions below $b$.  Then
\begin{align*}
Gg\left(\sum_{i=1}^b s_i\right) &= \frac{N\lambda}{k} P_A(\bm{s})\left(g\left(\sum_{i=1}^b s_i + \frac{k}{N}\right) - g\left(\sum_{i=1}^b s_i\right)\right)+ Ns_1\left(g\left(\sum_{i=1}^b s_i - \frac{1}{N}\right) - g\left(\sum_{i=1}^b s_i\right)\right),
\end{align*}
where first term is the drift due to a job arrival and the second term is due to a task departure.
To derive an upper bound on \eqref{eq:generator-approx}, we divide the discussion into the three cases below.  Recall that $g(x)=\frac{(x-\eta)^2}{2(-\delta)}\mathds{1}_{\{x\ge \eta\}}$ and $g'(x)=\frac{x-\eta}{-\delta}\mathds{1}_{\{x\ge \eta\}}$.

\noindent\textbf{Case 1}: $\sum_{i=1}^b S_i < \eta - \frac{k}{N}$. In this case, clearly $g'\left(\sum_{i=1}^b S_i\right) = 0$ and $Gg\left(\sum_{i=1}^b S_i\right) = 0$.

\noindent\textbf{Case 2}: $\sum_{i=1}^b S_i \in [\eta - \frac{k}{N}, \eta + \frac{1}{N})$. By the mean value theorem,
\begin{align}
g'\left(\sum_{i=1}^b S_i\right)(-\delta) - Gg\left(\sum_{i=1}^b S_i\right)
&= g'\left(\sum_{i=1}^b S_i\right)(-\delta) - \left(\frac{N\lambda}{k}P_A(\bm{S})\frac{k}{N}g'(\xi) + NS_1\frac{-1}{N}g'(\tilde{\xi})\right)\nonumber\\
&\le g'\left(\sum_{i=1}^b S_i\right)(-\delta) - \lambda g'(\xi) + S_1g'(\tilde{\xi}),\label{eq:case-2}
\end{align}
where $\xi \in \left(\sum_{i=1}^b S_i, \sum_{i = 1}^b S_i + \frac{k}{N}\right)$, $\tilde{\xi} \in \left(\sum_{i=1}^b S_i - \frac{1}{N}, \sum_{i = 1}^b S_i\right)$, and \eqref{eq:case-2} is true since $P_A(\bm{S})\le 1$ and $g'(x)\le 0$ for all $x$.

\noindent\textbf{Case 3}: $\sum_{i=1}^b S_i \geq \eta + \frac{1}{N}$. Since $g'(x)$ is continuous for all $x$, by the second order Taylor expansion in the Lagrange form,
\begin{align}
&\mspace{23mu}g'\left(\sum_{i=1}^b S_i\right)(-\delta) - Gg\left(\sum_{i=1}^b S_i\right)\nonumber\\
&= g'\left(\sum_{i=1}^b S_i\right)(-\delta) -\frac{N\lambda}{k}P_A(\bm{S})\left( \frac{k}{N}g'\left(\sum_{i=1}^b S_i\right) + \frac{k^2}{2N^2}g''(\zeta)\right)
-NS_1\left(\frac{-1}{N}g'\left(\sum_{i=1}^b S_i\right) + \frac{1}{2N^2}g''(\tilde{\zeta})\right)\nonumber\\
&\le g'\left(\sum_{i=1}^b S_i\right)\left(-\delta - \lambda + S_1\right) - \frac{1}{2N}\left(\lambda k g''(\zeta) + S_1g''(\tilde{\zeta})\right),
\end{align}
where $\zeta \in \left(\sum_{i=1}^b S_i, \sum_{i = 1}^b S_i + \frac{k}{N}\right)$, $\tilde{\zeta} \in \left(\sum_{i=1}^b S_i - \frac{1}{N}, \sum_{i = 1}^b S_i\right)$.

Combining these three cases yields
\begin{align}
&\mspace{23mu}\mathbb{E}\left[g'\left(\sum_{i=1}^b S_i\right)(-\delta) - Gg\left(\sum_{i=1}^b S_i\right)\right]\nonumber\\
&\le\mathbb{E}\left[\left(g'\left(\sum_{i=1}^b S_i\right)(-\delta)-\lambda g'(\xi) + S_1g'(\tilde{\xi})\right)\mathds{1}_{\left\{\sum_{i=1}^b S_i \in [\eta - \frac{k}{N}, \eta + \frac{1}{N})\right\}}\right]\label{termC}\\
&\mspace{21mu}-\frac{1}{2N}\mathbb{E}\left[(\lambda k g''(\zeta) + S_1g''(\tilde{\zeta}))\mathds{1}_{\left\{\sum_{i=1}^b S_i \geq \eta + \frac{1}{N}\right\}}\right] \label{termB}\\
&\mspace{21mu}+\mathbb{E}\left[g'\left(\sum_{i=1}^b S_i\right)(-\delta - \lambda + S_1) \mathds{1}_{\left\{\sum_{i=1}^b S_i \geq \eta + \frac{1}{N}\right\}}\right].\label{termA}
\end{align}

The first two terms \eqref{termC} and \eqref{termB} are easy to bound once we notice that for any $x \in \left[\eta - \frac{k + 1}{N},\eta + \frac{k + 1}{N}\right]$, $|g'(x)| \leq \frac{|x - \eta|}{\delta} \leq \frac{1}{\sqrt{N}\log N}$, and for any $x \in (\eta, +\infty)$, $|g''(x)| = \frac{1}{\delta} = \frac{\sqrt{N}}{(k + 1)\log N}$.
Then when $N$ is sufficiently large,
\begin{align}
|\eqref{termC}| &\leq \frac{1}{\sqrt{N}\log N}\left(\frac{(k+1)\log N}{\sqrt{N}}+1+1\right)\le\frac{3}{\sqrt{N}\log N}\nonumber,
\end{align}
and
\begin{align}
|(\ref{termB})| &\leq \frac{1}{2N}\frac{\sqrt{N}}{(k+1)\log N}\left(\lambda k + 1\right) \leq \frac{1}{\sqrt{N}\log N}.\nonumber
\end{align}

The \emph{key} in this proof is to bound the term \eqref{termA}, for which we utilize the state-space collapse result we establish in Lemma~\ref{lem:ssc} below.  The proof of Lemma~\ref{lem:ssc} is given in Appendix~\ref{app:lemma:ssc}.
\begin{lemma}[State-Space Collapse]\label{lem:ssc}
Under the assumption of Theorem \ref{main_theorem_1}, consider the following Lyapunov function:
$$
V(\bm{s}) = \min\left\{\frac{1}{h - 1}\sum_{i=h}^b s_i, b\left(\left(1 - \frac{1}{2}\beta N^{-\alpha}\right) - \frac{1}{h-1}\sum_{i=1}^{h - 1} s_i\right)^+\right\},
$$
where the superscript $^+$ denotes the function $x^+=\max\{x,0\}$.
Let $B = \frac{b - h + 1}{h - 1}\left(\beta N^{-\alpha} + \frac{\log N}{\sqrt{N}}\right)$. Then for any state $\bm{s}$ such that $V(\bm{s}) > B$, its Lyapunov drift can be upper bounded as follows
$$
\Delta V(\bm{s}) = GV(\bm{s}) \leq -\frac{b}{\sqrt{N}}.
$$
Consequently, when $N$ is sufficiently large,
\begin{align*}
\Pr\left\{V(\bm{S}) > B + \frac{2kb\log^2 N}{(h-1)\sqrt{N}}\right\}\le e^{-\frac{1}{2}\log^2 N}.
\end{align*}
\end{lemma}

With Lemma~\ref{lem:ssc}, we partition the probability space based on the value of $V(S)$ for bounding \eqref{termA}.  Note that $g'\left(\sum_{i=1}^b S_i\right)(-\delta - \lambda + S_1) \mathds{1}_{\left\{\sum_{i=1}^b S_i \geq \eta + \frac{1}{N}\right\}}$ is always no larger than $\frac{2b}{\delta}$ for large enough $N$.  Then \eqref{termA} can be upper bounded as:
\begin{align}
\eqref{termA}&\le \mathbb{E}\left[g'\left(\sum_{i=1}^b S_i\right)(-\delta - \lambda + S_1)
\cdot\mathds{1}_{\left\{\sum_{i=1}^b S_i \geq \eta + \frac{1}{N}\right\}}\;\middle|\; V(\bm{S}) \le B + \frac{2kb\log^2 N}{(h-1)\sqrt{N}}\right]\nonumber\\
&\mspace{21mu}+\frac{2b}{\delta}\Pr\left\{V(\bm{S}) > B + \frac{2kb\log^2 N}{(h-1)\sqrt{N}}\right\}.\label{eq:upper-termA}
\end{align}
Now we focus on the case where we are given the condition that $V(\bm{S}) \le B + \frac{2kb\log^2 N}{(h-1)\sqrt{N}}$.  Our goal is to show that $S_1$ is large enough such that $\delta+\lambda-S_1<0$.  Intuitively, this condition on $V(\bm{S})$ implies that we either have a small $\sum_{i=h}^b S_i$, which leads to a large $S_1$ when combined with the condition $\sum_{i=1}^b S_i \geq \eta + \frac{1}{N}$ in the indicator, or a large $\sum_{i=1}^{h - 1} S_i$, which directly gives a large $S_1$ since $S_1\geq \cdots\ge S_{h - 1}$.

\begin{sloppypar}
If $\frac{1}{h - 1}\sum_{i=h}^b S_i \le  b\left(\left(1 - \frac{1}{2}\beta N^{-\alpha}\right) - \frac{1}{h-1}\sum_{i=1}^{h - 1} S_i\right)^+$ in $V(\bm{S})$, the condition $V(\bm{S}) \le B + \frac{2kb\log^2 N}{(h-1)\sqrt{N}}$ implies that
\begin{equation}\label{eq:left-small}
\frac{1}{h - 1}\sum_{i=h}^{b} S_i \leq \frac{b - h + 1}{h - 1}\left(\beta N^{-\alpha} + \frac{\log N}{\sqrt{N}}\right) + \frac{2kb\log^2 N}{(h-1)\sqrt{N}}.
\end{equation}
Recall that $b = \min\left\{N^{\alpha},\frac{N^{0.5-\alpha}}{k}\right\}$ and $h = o(\log k)$.  Note that the indicator function in \eqref{eq:upper-termA} makes it sufficient to consider the case where $\sum_{i=1}^b S_i \geq \eta + \frac{1}{N}$, which implies $(h-1)S_1+\sum_{i=h}^b S_i \geq \eta$.  Combining this with \eqref{eq:left-small} gives
\begin{align*}
S_1&\ge\frac{\eta}{h-1}-\frac{b - h + 1}{h - 1}\left(\beta N^{-\alpha} + \frac{\log N}{\sqrt{N}}\right) - \frac{2kb\log^2 N}{(h-1)\sqrt{N}}\\
&\ge 1 + (1-\beta)\frac{1}{h-1}-\frac{1}{2}\beta N^{-\alpha} + o\left(\frac{1}{h}\right)
\end{align*}
when $N$ is sufficiently large.  Note that $\delta=o\left(\frac{1}{h}\right)$ and $\lambda=1-\beta N^{-\alpha}$.  Therefore, $\lambda + \delta - S_1 < 0$ when $N$ is sufficiently large.
\end{sloppypar}

\begin{sloppypar}
If $\frac{1}{h - 1}\sum_{i=h}^b S_i >  b\left(\left(1 - \frac{1}{2}\beta N^{-\alpha}\right) - \frac{1}{h-1}\sum_{i=1}^{h - 1} S_i\right)^+$ in $V(\bm{S})$, the condition $V(\bm{S}) \le B + \frac{2kb\log^2 N}{(h-1)\sqrt{N}}$ implies that
\begin{align*}
b&\left(1 - \frac{1}{2}\beta N^{-\alpha}-\frac{1}{h-1}\sum_{i=1}^{h-1} S_i\right) \le B + \frac{2kb\log^2 N}{(h-1)\sqrt{N}}.
\end{align*}
Then
\begin{align*}
S_1 &\ge \frac{1}{h-1}\sum_{i=1}^{h-1} S_i\\
&\ge 1 - \frac{1}{2}\beta N^{-\alpha} - \frac{1}{b}\left(B + \frac{2kb\log^2 N}{(h-1)\sqrt{N}}\right) \\
&\geq 1 - \frac{1}{2}\beta N^{-\alpha} + o(N^{-\alpha}).
\end{align*}
As a result, again we have $\lambda + \delta - S_1 \leq -\frac{1}{2}\beta N^{-\alpha} + o(N^{-\alpha}) < 0$ when $N$ is sufficiently large.
\end{sloppypar}

Inserting these bounds back to \eqref{eq:upper-termA} gives that when $N$ is sufficiently large,
\begin{align*}
\eqref{termA} &\le 0 + \frac{2b}{\delta}\Pr\left\{V(\bm{S}) > B + \frac{2kb\log^2 N}{(h-1)\sqrt{N}}\right\}\\
&\le \frac{2b}{\delta}e^{-\frac{1}{2}\log^2 N}\\
&\leq \frac{1}{\sqrt{N}\log N}.
\end{align*}

Combining the bounds for \eqref{termC}, \eqref{termB} and \eqref{termA}, we have
$$
\mathbb{E}\left[\max\left\{\sum_{i=1}^b S_i - h\left(1-\frac{1}{2}\beta N^{-\alpha}\right), 0\right\}\right] \le \frac{5}{\sqrt{N}\log N},
$$
which completes the proof of Theorem \ref{main_theorem_1}.
\end{proof}

\subsection{Proof of Theorem \ref{theorem zero delay}}\label{sec:proof-theorem zero delay}
\begin{proof}
We first bound the dropping probability $\probdrop$ using Lemma~\ref{lem:filling} with the threshold value $\ell=b$.  Note that an incoming job does not get dropped if and only if all its $k$ tasks are routed to queueing positions below threshold $b$, which is the complement of the event $\textrm{FILL}_b$ in Lemma~\ref{lem:filling}. Thus,
\begin{align*}
\probdrop&=1-\Pr\{\textrm{FILL}_b\}\\
&=1-\Pr\left\{\textrm{FILL}_b\;\middle|\; \sum_{i=1}^b S_i \le b\left(1 - \frac{1}{4}\beta N^{-\alpha}\right)\right\}
\cdot\Pr\left\{\sum_{i=1}^b S_i \le b\left(1 - \frac{1}{4}\beta N^{-\alpha}\right)\right\}\\
&\mspace{23mu}-\Pr\left\{\textrm{FILL}_b\;\middle|\; \sum_{i=1}^b S_i > b\left(1 - \frac{1}{4}\beta N^{-\alpha}\right)\right\}
\cdot\Pr\left\{\sum_{i=1}^b S_i > b\left(1 - \frac{1}{4}\beta N^{-\alpha}\right)\right\}.
\end{align*}
We can easily have that $\Pr\left\{\textrm{FILL}_b\;\middle|\; \sum_{i=1}^b S_i \le b\left(1 - \frac{1}{4}\beta N^{-\alpha}\right)\right\}\le \frac{1}{N}$ using Lemma~\ref{lem:filling}.

Now we bound $\Pr\left\{\sum_{i=1}^b S_i > b\left(1 - \frac{1}{4}\beta N^{-\alpha}\right)\right\}$ using Theorem~\ref{main_theorem_1}.  Note that
\begin{align*}
&\mspace{23mu}\Pr\left\{\sum_{i=1}^b S_i > b\left(1 - \frac{1}{4}\beta N^{-\alpha}\right)\right\}\\
&\le \Pr\left\{\max\left\{\sum_{i=1}^b S_i - h\left(1 - \frac{1}{2}\beta N^{-\alpha}\right),0\right\}> b - \frac{b}{4}\beta N^{-\alpha}-h\right\}\\
&\le \Pr\left\{\max\left\{\sum_{i=1}^b S_i - h\left(1 - \frac{1}{2}\beta N^{-\alpha}\right),0\right\}> \frac{b}{2}\right\},
\end{align*}
where we have used the fact that $\frac{b}{4}\beta N^{-\alpha}+h\le \frac{b}{2}$ when $N$ is sufficiently large due to our assumptions on $b$ and $h$.  Then by Markov's inequality,
\begin{align*}
\Pr\left\{\sum_{i=1}^b S_i > b\left(1 - \frac{1}{4}\beta N^{-\alpha}\right)\right\}
&\le\frac{\expect\left[\max\left\{\sum_{i=1}^b S_i - h\left(1 - \frac{1}{2}\beta N^{-\alpha}\right),0\right\}\right]}{\frac{b}{2}}\\
&\le \frac{10}{b\sqrt{N}\log N}.
\end{align*}

Combining the arguments above yields
\begin{align*}
\probdrop\ge 1-\frac{1}{N}-\frac{10}{b\sqrt{N}\log N}\ge 1-\frac{11}{b\sqrt{N}\log N}
\end{align*}
when $N$ is sufficiently large.

Next we bound the expected job delay given that a job is admitted, i.e., $\expect[T\;|\;\textup{admitted}]$.  We define the delay of a job that is dropped to be zero since it leaves the system immediately after arrival.  Then $\expect[T]=\expect[T\;|\;\textup{admitted}]\cdot(1-\probdrop)+\expect[T\;|\;\textrm{dropped}]\cdot\probdrop$, and thus $\expect[T\;|\;\textup{admitted}]=\frac{\expect[T]}{1-\probdrop}$.  So we can focus on bounding $\expect[T]$, following the outline given in the proof sketch.

We bound $\expect[T]$ in the following way
\begin{align}
\mathbb{E}[T] &\leq \mathbb{E}\left[T \;\middle|\; \sum_{i=1}^h S_i \leq h\left(1 - \frac{1}{4}\beta N^{-\alpha}\right)\right]\label{eq:term-fill}\\ 
&\mspace{21mu}+ \mathbb{E}\left[T \;\middle|\; \sum_{i=1}^h S_i > h\left(1 - \frac{1}{4}\beta N^{-\alpha}\right)\right]
\cdot\Pr\left\{\sum_{i=1}^h S_i > h\left(1 - \frac{1}{4}\beta N^{-\alpha}\right)\right\}.\label{eq:term-not-fill}
\end{align}

For the first term \eqref{eq:term-fill} in this upper bound, as described in the proof sketch, we will rely on the fact that with high probability, all the $k$ tasks are assigned to queueing positions below $h$.  Specifically,
\begin{align*}
&\mspace{22mu}\mathbb{E}\left[T \;\middle|\; \sum_{i=1}^h S_i \leq h\left(1 - \frac{1}{4}\beta N^{-\alpha}\right)\right]\\
&=\mathbb{E}\left[T \;\middle|\; \sum_{i=1}^h S_i \leq h\left(1 - \frac{1}{4}\beta N^{-\alpha}\right),\textrm{FILL}_h\right]
\cdot\Pr\left\{\textrm{FILL}_h\;\middle|\;\sum_{i=1}^h S_i \leq h\left(1 - \frac{1}{4}\beta N^{-\alpha}\right)\right\}\\
&\mspace{23mu}+\mathbb{E}\left[T \;\middle|\; \sum_{i=1}^h S_i \leq h\left(1 - \frac{1}{4}\beta N^{-\alpha}\right),\overline{\textrm{FILL}_h}\right]
\cdot\Pr\left\{\overline{\textrm{FILL}_h}\;\middle|\;\sum_{i=1}^h S_i \leq h\left(1 - \frac{1}{4}\beta N^{-\alpha}\right)\right\},
\end{align*}
where $\overline{\textrm{FILL}_h}$ is the complement of $\textrm{FILL}_h$.

Suppose $\textrm{FILL}_h$ is true.  Suppose that the $k$ tasks of the incoming job land in $m$ distinct queues with $m\le k$.  We call the tasks with the highest positions in these $m$ queues tasks $1,2,\dots,m$, and let $n_1,n_2,\dots,n_m$ denote these positions.  Then the delay of task $i$ can be written as $Y_i=\sum_{j=1}^{n_i}X_{i,j}$, where $X_{i,j}$ is the service time of the task at position $j$ in the same queue as task $i$.  Clearly $X_{i,j}$'s are i.i.d.\ with an exponential distribution of rate $1$.  We know that $n_i\le h, i=1,2,\dots,m$ given $\textrm{FILL}_h$.  Then by Lemma~\ref{lem:max-upper},
$$
\mathbb{E}[\max\left\{Y_1,\cdots,Y_m\right\}] \leq \ln k + o(\ln k).
$$

\begin{sloppypar}
When $\overline{\textrm{FILL}_h}$ is true, $\mathbb{E}\left[T \;\middle|\; \sum_{i=1}^h S_i \leq h\left(1 - \frac{1}{4}\beta N^{-\alpha}\right),\overline{\textrm{FILL}_h}\right]\le bk$ since the highest position for a task is $b$ and the maximum is upper bounded by the sum.  Further, $\Pr\left\{\overline{\textrm{FILL}_h}\;\middle|\;\sum_{i=1}^h S_i \leq h\left(1 - \frac{1}{4}\beta N^{-\alpha}\right)\right\}\le \frac{1}{N}$ by Lemma~\ref{lem:filling}.
\end{sloppypar}

Combining the arguments above, we have the following bound for term \eqref{eq:term-fill}:
\begin{align*}
\mathbb{E}\left[T \;\middle|\; \sum_{i=1}^h S_i \leq h\left(1 - \frac{1}{4}\beta N^{-\alpha}\right)\right]\le \ln k+o(\ln k)+\frac{bk}{N}.
\end{align*}

Now we go back to the term \eqref{eq:term-not-fill}.  Again, it is easy to see that $\mathbb{E}\left[T \;\middle|\; \sum_{i=1}^h S_i > h\left(1 - \frac{1}{4}\beta N^{-\alpha}\right)\right]\le bk$.  Utilizing Theorem~\ref{main_theorem_1}, we have
\begin{align*}
\Pr\left\{\sum_{i=1}^h S_i > h\left(1 - \frac{1}{4}\beta N^{-\alpha}\right)\right\}
&\le \Pr\left\{\max\left\{\sum_{i=1}^b S_i-h\left(1-\frac{1}{2}\beta N^{-\alpha}\right),0\right\} > \frac{1}{4}h\beta N^{-\alpha}\right\}\\
&\le\frac{\expect\left[\max\left\{\sum_{i=1}^b S_i-h\left(1-\frac{1}{2}\beta N^{-\alpha}\right),0\right\}\right]}{\frac{1}{4}h\beta N^{-\alpha}}\\
&\le \frac{20}{h\beta N^{\frac{1}{2}-\alpha}\log N}.
\end{align*}

With the bounds above on \eqref{eq:term-fill} and \eqref{eq:term-not-fill}, we have
\begin{equation*}
\expect[T]\le \ln k+o(\ln k)+\frac{bk}{N}+\frac{20bk}{h\beta N^{\frac{1}{2}-\alpha}\log N}.
\end{equation*}
Consequently,
\begin{align*}
\expect[T\;|\; \text{admitted}]&=\frac{\expect[T]}{1-\probdrop}\\
&\le \frac{\ln k+o(\ln k)+\frac{bk}{N}+\frac{20bk}{h\beta N^{\frac{1}{2}-\alpha}\log N}}{{1-\probdrop}}\\
&\le \ln k + o(\ln k),
\end{align*}
which completes the proof.
\end{proof}

\section{Proofs for Impossibility Results (Theorems~\ref{main_theorem_2} and \ref{theorem lower d})}\label{sec:impossibility}
In this section, we prove the impossibility results in Theorems~\ref{main_theorem_2} and \ref{theorem lower d}.  Throughout this section, we assume that the assumptions in Theorem~\ref{main_theorem_2} hold true. Due to the space limit, the lemmas needed and their proofs are presented in Appendix~\ref{app:lemmas-impossibility}.

{
\subsubsection*{\textbf{Proof sketch}}
We focus on proving the lower bound in Theorem~\ref{main_theorem_2} since the non-zero queueing delay result in Theorem~\ref{theorem lower d} follows from that rather straightforwardly.

Our proof of Theorem~\ref{main_theorem_2} uses a novel lower bounding technique we develop.  We derive the lower bound on $\Pr\left\{\sum_{i=1}^h S_i \geq h - \frac{1}{3d}\right\}$ by lower-bounding $\Pr\left\{S_1 - S_h \leq c_h\right\}$ for a properly chosen $c_h$, for which our proof proceeds in an inductive fashion.
\begin{itemize}[leftmargin=1em,topsep=0.5em]
    \item We first lower bound $\Pr\left\{S_1 - S_2 \leq c_2\right\}$ utilizing a tail bound for $S_1$, which can be easily obtained from Little's law.  This step uses Lyapunov-based tail bounds in a ``reverse'' way in the following sense.  Typically, one can analyze the terms in the Lyapunov drift to obtain a tail bound.  But here, we utilize a tail bound obtained through other ways to bound a term (the probability in Lemma~\ref{lem:tail}) in the Lyapunov drift.
    \item We then lower bound $\Pr\left\{S_1 - S_3 \leq c_3\right\}$ based on the lower bound on $\Pr\left\{S_1 - S_2 \leq c_2\right\}$ in the previous step following a similar argument.  We continue this procedure inductively until we get the desired lower bound on $\Pr\left\{S_1 - S_h \leq c_h\right\}$.
\end{itemize}
}

\subsection{Proof Of Theorem \ref{main_theorem_2}}\label{sec:proof-main_theorem_2}
\begin{proof}
As outlined in the proof sketch, we first lower-bound $\Pr\left\{S_1 - S_h \leq c_h\right\}$ using arguments in an inductive fashion.  We start by lower-bounding $\Pr\left\{S_1 - S_2 \leq c_2\right\}$ for a properly chosen $c_2$.  This base case relies on the fact that $\mathbb{E}[S_1] = 1 - \beta N^{-\alpha}$, which can be easily proven using Little's law.

To simplify notation, let $u = 2kd$. Consider the Lyapunov function 
$
V_1(\bm{s}) = s_1.
$
Let $h=O(\log k)$ and $B_1 = 1 - h\beta N^{-\alpha}$. For some state $\bm{s}$ such that $V_1(\bm{s}) > B_1$, it holds that
$$
\begin{aligned}
\Delta V_1(\bm{s}) &=  \sum_{\bm{s}':\bm{s} \to \bm{s}' \text{~due to an arrival}} r_{\bm{s} \to \bm{s}'}\left(V_{1}(\bm{s}') - V_1(\bm{s})\right)
+\sum_{\bm{s}':\bm{s} \to \bm{s}' \text{~due to a departure}} r_{\bm{s} \to \bm{s}'}\left(V_{1}(\bm{s}') - V_1(\bm{s})\right) \\
&\stackrel{\text{(a)}}{\le} uh\beta N^{-\alpha}-N(s_1-s_2)\frac{1}{N} \\
&= uh\beta N^{-\alpha}-(s_1-s_2),
\end{aligned}
$$
where (a) is due to Lemma~\ref{lower:lemma1}.

Consider the set of states $\mathcal{E}_1 = \left\{\bm{s} \in \mathcal{S} | s_1 - s_2 > uh^2\beta N^{-\alpha}\right\}$.  Let  
$p_2=\Pr\left\{\bm{S} \not \in \mathcal{E}_1\right\}$, which is equal to $\Pr\left\{S_1 - S_2 \le uh^2\beta N^{-\alpha}\right\}$.
We now use the tail bound in Lemma \ref{lem:tail}. Assume that we follow the notation in the lemma. Consider the following two cases:
\begin{itemize}
\item $\bm{s} \not \in \mathcal{E}_1$, $\Delta V_1(\bm{s}) \leq uh\beta N^{-\alpha} \eqqcolon \delta$.
\item $\bm{s} \in \mathcal{E}_1$. Let $\gamma = -\Delta V_1(\bm{s})$. It holds $\gamma \geq uh\beta N^{-\alpha}(h - 1)$.
\end{itemize}
Following the definition in \ref{lem:tail}, it is easy to verify that $\nu_{\mathrm{max}} \leq \frac{k}{N}$ and $f_{\mathrm{max}} \leq 1$ for $V_1(\bm{s})$. Let $j_1 = \left(1 + \frac{N^{\alpha}}{\beta uh(h-1)}\right)\log^2 N$. By Lemma \ref{lem:tail}, it holds that 
$$
\begin{aligned}
\Pr\left\{V_1(\bm{S}) > B_1 + 2\nu_{\mathrm{max}}j_1\right\}& \leq \left(\frac{f_{\mathrm{max}}}{f_{\mathrm{max}} + \gamma}\right)^{j_1} + \left(\frac{\delta}{\gamma} + 1\right)\Pr\left\{\bm{S} \not \in \mathcal{E}_1\right\} \\
&\leq \left(\frac{f_{\mathrm{max}}}{f_{\mathrm{max}} + \gamma}\right)^{j_1} + \frac{h}{h-1}p_2.
\end{aligned}
$$
Note that when $N$ is large enough, $\left(\frac{f_{\mathrm{max}}}{f_{\mathrm{max}} + \gamma}\right)^{j_1} \leq \left(1 + uh\beta N^{-\alpha}(h-1)\right)^{-\left(1 + N^{\alpha}\frac{1}{\beta uh(h-1)}\right)\log^2 N}\leq e^{-\log^2 N}.$
As a result,
$$
\Pr\left\{V_1(\bm{S}) > B_1 + 2\nu_{\mathrm{max}}j_1\right\} \leq N^{-\log N} + \frac{h}{h-1}p_2.
$$
Since $0 < \alpha < 0.5$ and $k = e^{o\left(\sqrt{\log N}\right)}$, we have
$
B_1 + 2\nu_{\mathrm{max}}j_1=1 - h\beta N^{-\alpha} + 2 \frac{k}{N}\left(1 + \frac{N^{\alpha}}{\beta uh(h-1)}\right)\log^2 N<1 - (h - 1)\beta N^{-\alpha}
$
when $N$ is large enough.
It then follows that
$$
\begin{aligned}
\Pr\left\{V_1(\bm{S}) > 1 - (h - 1)\beta N^{-\alpha}\right\} &\leq \Pr\left\{V_1(\bm{S}) > B_1 + 2\nu_{\mathrm{max}}j_1\right\} 
\leq N^{-\log N} + \frac{h}{h-1}p_2.
\end{aligned}
$$

We now combine the bound above with the following bound given by Lemma \ref{lower:lemma0}: 
$$
\Pr\left\{V_1(\bm{S}) > 1 - (h - 1)\beta N^{-\alpha}\right\}
\geq 1 - \frac{1}{h-1}.
$$
Therefore, 
$
\frac{h}{h-1}p_2+N^{-\log N}\geq \frac{h-2}{h-1},
$
and thus
$$
\Pr\left\{S_1 - S_2 \leq uh^2\beta N^{-\alpha}\right\} = p_2 \geq \frac{h-2}{h}-N^{-\log N}.
$$
Let $b_q = u^{q-1}h^q\beta N^{-\alpha}$ for an integer $q > 0$. Define a sequence $a_q$, such that $a_1 = 0,a_2 = 1$ and $a_q = (q-2)a_{q-1} + 2$ for $q > 2$.  With this notation, the lower bound above on $p_2$ can be rewritten as
$
\Pr\left\{S_1 - S_2 \leq a_2b_2\right\} \geq \frac{h-2}{h}-N^{-\log N}.
$
We can use Lemma \ref{lower:lemma2} inductively to show that for all $q$ with $2 \leq q \leq h$,
$$
\Pr\left\{S_1 - S_{q} \leq a_qb_q\right\} \geq \left(\frac{h-2}{h}\right)^{q-1}-(q-1)N^{-\log N}.
$$

Let us condition on $S_1 - S_h \leq a_hb_h$. For ease of notation, let $\probcond = \left(\frac{h-2}{h}\right)^{h-1}-(h-1)N^{-\log N}$, which is a lower bound on the probability of the condition. Note that $$
\begin{aligned}
\expect[S_1] &\leq \expect\left[S_1 \;\middle| \;S_1 - S_h \le a_hb_h\right]\cdot\Pr\left\{S_1 - S_h \leq a_hb_h\right\} 
+ 1 \cdot \Pr\left\{S_1 - S_h > a_hb_h\right\}.
\end{aligned}
$$ 
Thus
$$
\begin{aligned}
\expect\left[S_1 \;\middle|\; S_1 - S_h \leq a_hb_h\right] &\geq \frac{1 - \beta N^{-\alpha} - \left(1 - \Pr\left\{S_1 - S_h \leq a_hb_h\right\}\right)}{\Pr\left\{S_1 - S_h \leq a_hb_h\right\}} \\
&\geq 1 - \frac{\beta}{\probcond}N^{-\alpha}.
\end{aligned}
$$
We can also see that
\begin{equation}\label{eq:lower_eq1}
\begin{aligned}
&\mspace{21mu}\Pr\left\{\sum_{i=1}^h S_i \geq h - \frac{1}{3d}\right\} \\
&\geq  \Pr\left\{\sum_{i=1}^h S_i \geq h - \frac{1}{3d} \; \middle | \; S_1 - S_h \leq a_hb_h\right\}\Pr\left\{S_1 - S_h \leq a_hb_h\right\}\\
&\geq  \probcond \Pr\left\{hS_1 - h(S_1 - S_h) \geq h - \frac{1}{3d} \; \middle | \; S_1 - S_h \leq a_hb_h\right\} \\
&\geq  \probcond \Pr\left\{S_1 \geq 1 - \frac{1}{3dh} + a_hb_h \; \middle | \; S_1 - S_h \leq a_hb_h\right\}.
\end{aligned}
\end{equation}
Utilizing the Markov inequality gives
$$
\begin{aligned}
(\ref{eq:lower_eq1}) &\geq \probcond \left(1 - \frac{3dh - 3dh\expect\left[S_1 \;\middle | \; S_1 - S_h \leq a_hb_h\right]}{1 - 3dh a_hb_h}\right) \\
&\geq \probcond \left(1 - \frac{\beta}{\probcond}\frac{3dh}{1 - 3dh a_hb_h}N^{-\alpha}\right).
\end{aligned}
$$
Recall that $a_q = (q - 2)a_{q-1} + 2$ for $q > 2$ and $a_2 = 1$. We have $a_h \leq 2h^h$, and thus $a_hb_h \leq 2\beta u^hh^{2h}N^{-\alpha}$. As $d = e^{o(\log N / \log k)}, k = e^{o(\sqrt{\log N})}, h = O(\log k)$, we have $\ln(a_hb_h) = -\Omega(\log N)$. Furthermore, since $\ln(3dh) = o(\log N / \log k) + O(\log k), \alpha > 0$, it holds 
$$
1 - \frac{\beta}{\probcond}\frac{3dh}{1 - 3dh a_hb_h}N^{-\alpha} \geq \frac{1}{2}
$$
if $N$ is sufficiently large. Note that $\probcond$ is equal to $ \left(\frac{h-2}{h}\right)^{h-1}-(h-1)N^{-\log N}$ which converges to $\frac{1}{e^2}$. We could conclude that when $N$ goes to infinity, we have
$$
\Pr\left\{\sum_{i=1}^h S_i \geq h - \frac{1}{3d}\right\} \geq \frac{1}{4e^2}.
$$
\end{proof}
\subsection{Proof Of Theorem \ref{theorem lower d}}\label{sec:proof-theorem_lower_d}
\begin{proof}
Let $h = 12e^2\ln k$. Then $h = O(\log k)$. Suppose that we have an incoming job. By Theorem \ref{main_theorem_2} and the PASTA property of a Poisson arrival process, with probability at least $\frac{1}{4e^2}$, this job will see a state $\bm{s}$ such that 
$
\sum_{i=1}^h s_i \geq h - \frac{1}{3d}.
$
By Lemma \ref{lower:lemma3}, the dispatcher will route at least one task of this job into a queue of length at least $h + 1$ with probability $1 - o(1)$. Let $T$ be the delay of the job. Then it holds for a large enough $N$, 
$$
\mathbb{E}[T] \geq 3\ln k(1 - o(1)) \geq 2\ln k,
$$
which completes the proof. 
\end{proof}

\section{Discussion on an alternative notion of zero queueing delay}\label{sec:alter}
{
In this section, we consider an alternative notion of zero queueing delay that may be of interest and may provide more understanding into the dynamics of systems with parallel jobs.  We will refer to this alternative notion as \emph{zero waiting} to differentiate it from the zero queueing delay we consider in the main part of the paper.  We say that zero waiting is achieved if in steady state, all the tasks of an incoming job enter service immediately upon arrival without waiting in queues with high probability as $N\to\infty$.  It is easy to see that zero waiting is a much stronger requirement than zero queueing delay.  Indeed, we show in Theorem~\ref{thm:lower-zero-waiting} below, the minimum probe overhead needed for achieving zero waiting is larger than $\frac{1}{2(1-\lambda)}$, which is in the same order as the value in the impossibility results for non-parallel jobs.  The proof of Theorem~\ref{thm:lower-zero-waiting} is straightforward and given in Appendix~\ref{sec:proof-thm-zero-waiting}.

Note that although this notion of zero waiting for parallel jobs seems to resemble the zero queueing delay for non-parallel jobs, the two systems have fundamentally different dynamics and thus it is hard to directly compare these two notion.  For parallel jobs, a batch of tasks arrive together and zero waiting requires all of them to be assigned to idle servers \emph{simultaneously}.  In contrast, for non-parallel jobs, there is no concept of batches.  The single-task jobs arrive one by one and zero queueing delay requires a job to be assigned to an idle server when it arrives.

\begin{theorem}\label{thm:lower-zero-waiting}
Consider a system with $N$ servers where each job consists of $k$ tasks. Let the load be $\lambda = 1 - \beta N^{-\alpha}$ with $0 < \beta \le 1$ and $\alpha \geq 0$. Assume that the buffers have unlimited sizes.  Under the batch-filling policy with a probe overhead $d$ such that $1 \leq d \leq \frac{1}{2(1-\lambda)}$, the probability in steady state that all the tasks of an incoming job are assigned to idle servers is smaller than or equal to $0.5$.
\end{theorem}
}

\section{Simulation Results}\label{sec:simulation}

{
In this section, we perform two sets of simulations to demonstrate our theoretical results and explore settings beyond those in our theoretical analysis.  The first set illustrates the scaling behavior of the system as $N$ grows under various probe ratios, and investigates the gap between our achievability results and impossibility results.  The second set of simulations experiment on more general service time distributions beyond the exponential distribution and correlation among task service times.


\subsection{Scaling Behavior with Various Probe Ratios}\label{subsec:scaling}
This set of simulations use the setting of our theoretical results with $\lambda=1-N^{-0.3}$ ($\alpha=0.3$ and $\beta=1$). We let $k$, the number of tasks per job, scale with $N$ as $k = \lfloor \ln^2 N \rfloor$.  The values for $N$ and the corresponding $k$ used in the simulations are given in Table~\ref{table:scale-system}. These values are reasonable in practice considering that datacenters nowadays typically have tens of thousands of nodes (each with multiple cores) per cluster and a job may consist of hundreds of tasks \cite{amvrosiadis_2018}.

\begin{table}[!hbp]
\begin{tabular}{r r r r r r r r r r r r r}
\toprule
$N$ & 32   & 64   & 128  & 256   & 512   & 1024  & 2048 & 4096 & 8192 & 16384 & 32768 & 65536 \\ 
$k$ & 12   & 17   & 23   & 30    & 38    & 48   & 58   & 69   & 81   & 94    & 108   & 122   \\
\bottomrule
\end{tabular}
\caption{Scaling parameters}
\label{table:scale-system}
\vspace{-0.2in}
\end{table}

\begin{figure}
\begin{minipage}[t]{.45\textwidth}
\centering
\scalebox{0.32}{\input{picture/diff-d.pgf}}
\vspace{-0.2in}
\caption{Queueing delays under different probe ratios: $d_1$ is sufficient for convergence to zero queueing delay; $d_1>d_4>d_3>d_2$.}
\label{fig:different-d}
\end{minipage}
\hspace{0.3in}
\begin{minipage}[t]{.45\textwidth}
    \centering
    \scalebox{0.32}{\input{picture/diff-dist.pgf}}
    \vspace{-0.2in}
    \caption{Queueing delays under different service time distributions.}
    \label{fig:different-dist}
\end{minipage}
\vspace{-0.1in}
\end{figure}

We explore four scaling settings of the probe ratio.  The first setting uses a probe ratio of $d_1 = \frac{4}{(1-\lambda)h}=\frac{4N^{\alpha}}{h}$ with $h = \lceil\frac{\log k}{\log \log k} \rceil$, which satisfies the conditions in Theorems~\ref{main_theorem_1} and \ref{theorem zero delay} to achieve zero queueing delay.  The second setting uses a probe ratio of $d_2 = \exp{(0.5\log N / \log k)}$, which is slightly larger than the value in Theorems~\ref{main_theorem_2} and \ref{theorem lower d} that guarantees non-zero queueing delay.  The other two settings use probe ratio values $d_3$ and $d_4$ that interpolate between $d_1$ and $d_2$ to investigate the threshold under which the system transits from zero queueing delay to non-zero queueing delay.  Note that $d_1>d_4>d_3>d_2$ for all values of $N$ in the simulations. More details of the settings can be found in Appendix~\ref{sec:experiment}.

Figure~\ref{fig:different-d} shows the simulation results for the queueing delay $\frac{\expect[T-T^*]}{\expect[T^*]}$, 
{ where the results are averaged over ten independent runs.  Please refer to Appendix~\ref{sec:exactval} for the exact values and standard deviations.  Since the standard deviations are very small ($\sim 10^{-4}$), the error bars are not visible in the plots.
}
The curve for $d_1$ demonstrates the trend for the queueing delay to converge to zero as predicted by the theoretical results.  It does not exactly reach zero but becomes reasonably close. Under the probe ratios $d_2$ and $d_3$, the queueing delay clearly deviates from zero.  Under the probe ratio $d_4$, the queueing delay flattens out after some initial drop as $N$ becomes large.  Therefore, it is plausible that the transition from zero queueing delay to non-zero queueing delay happens at a probe ratio value near $d_4$.  Since $d_4$ is much closer to $d_1$ than to $d_2$, we expect our impossibility results to be not tight.  Pinning down the exact threshold for the transition (or proving the nonexistence of such a threshold) is of great theoretical interest and we leave it for future research.

{
To further investigate how $k$, the number of tasks per job, affects the scaling behavior under different probe ratios, we examine another setting where $k = \lfloor \sqrt{N} \rfloor$.  This scaling of $k$ is beyond our theoretical framework, but the queueing delays exhibit similar trends as those in the setting where $k = \lfloor \ln^2 N \rfloor$.  Details of the simulation results are given in Appendix~\ref{sec:newk}.}

{
\subsection{More General Settings for Task Service Times}
This set of simulations explore distributions beyond the exponential distribution for task service times and correlation among their service times.  Figure~\ref{fig:different-dist} shows the results for four settings: (1)~\emph{i.i.d.\ exponential distribution} with rate $1$ (denoted as $\text{Exp}(1)$).  This is the baseline distribution that is assumed for our theoretical analysis. (2)~\emph{i.i.d.\ bounded Pareto} in range $[1,1000]$ with a shape constant $1.5$. (3)~\emph{i.i.d.\ hyper-exponential} that follows $\text{Exp}(1)$ with probability $0.99$ and $\text{Exp}(0.01)$ with probability $0.01$.  We re-scale the arrival rates so all the systems have the same load $\lambda=1-N^{-0.3}$.
(4)~\emph{S\&X model} for correlated task service times, which is a model proposed in \cite{GarHarSch_16} and has been extensively studied since then.  In the S\&X model, the service time of the each task in a job can be written as $S\cdot X$, where every task in the same job shares the same $X$, but different tasks have their own $S$'s that are independent among tasks. Here we assume $S$ and $X$ are both exponentially distributed with rate $1$.} The probe overhead is chosen to be the $d_1$ in Section~\ref{subsec:scaling} such that zero queueing delay is provably achievable under the exponential distribution. 

{
We observe that empirically, the queueing delay has a trend that approaches zero under all the four settings, despite of the larger coefficients of variation for the bounded Pareto and hyper-exponential distributions and the correlation among task service times in the S\&X model.  These simulation results suggest that our theoretical results have some robustness with respect to service time distributions and correlations.  We comment that there is little existing work on zero queueing delay for general service time distributions with the exception of \cite{b5}, which studies the Coxian-2 distribution for non-parallel jobs.  Generalizing our analysis to general service time distributions with possible correlations is a research direction that deserves much further effort, as it is for many problems in queueing systems.}
}

\section{Conclusions}
We studied queueing delay in a system where jobs consist of parallel tasks.  We first proposed a notion of zero queueing delay in a relative sense for such parallel jobs. We then derived conditions on the probe overhead for achieving zero queueing delay and for guaranteeing non-zero queueing delay.  One interesting implication of the results is that under parallelization, the probe overhead needed for achieving zero queueing delay is lower than that in a system with non-parallel (single-task) jobs under the same load.  Through simulations, we demonstrated that the numerical results are consistent with the theoretical results under reasonable settings, and investigated several questions that are hard to answer analytically.

\paragraph{Acknowledgment:} The work of Wentao Weng was conducted during a visit to the Computer Science Deparment, CMU in 2019.

\bibliographystyle{ACM-Reference-Format}
\bibliography{references,refs-weina}

\appendix

\section{Proofs of Lemmas~\ref{lem:max-upper}, \ref{lem:filling} and \ref{lem:ssc}}\label{app:lemmas-achievability}

\subsection{Proof of Lemma~\ref{lem:max-upper}}\label{app:lem:max-upper}
\renewcommand{\thmnamerestated}{Lemma~\ref{lem:max-upper} [Restated]}
\newtheorem*{thmrestatedmax}{\thmnamerestated}{\bfseries}{\itshape}
\begin{thmrestatedmax}
Consider $m$ independent random variables $Y_1,\cdots,Y_m$ with $m\le k$, where each $Y_i$ ($1\le i\le m)$ is the sum of $n_i$ i.i.d.\ random variables that follow the exponential distribution with rate $1$.  In the asymptotic regime that $k$ goes to infinity, if $\max\left\{n_1,\cdots,n_m\right\}=o(\log k)$, then
$$
\mathbb{E}[\max\left\{Y_1,\cdots,Y_m\right\}] \leq \ln k + o(\ln k).
$$
\end{thmrestatedmax}

\begin{proof}
The general proof idea is folklore, but here we derive the exact bounds for our purpose.
Let $M_{X}(s)$ be the moment generating function of a random variable $X$. By assumption, $Y_i = \sum_{j=1}^{n_i} X_{i,j}$, and $X_{i,j}, 1 \leq i \leq m, 1 \leq j \leq n_i$ are all independent and exponentially distributed with mean $1$. Therefore, for any $1 \leq i \leq m, 1 \leq j \leq n_i$ and any $s$ with $0<s<1$,
$$
\begin{aligned}
M_{X_{i,j}}(s) &= \mathbb{E}\left[e^{sX_{i,j}}\right] = \frac{1}{1 - s}, \\
M_{Y_i}(s) &= \mathbb{E}\left[e^{sY_i}\right]=\left(\frac{1}{1 - s}\right)^{n_i}.
\end{aligned}
$$

Let $q = \max\left\{n_1,\cdots,n_m\right\}$. It holds that for any $s$ with $0<s<1$, 
\begin{align}
\exp\left(s\mathbb{E}\left[\max_{j=1}^m Y_j\right]\right) &\leq \mathbb{E}\left[\exp(s\max_{j=1}^m Y_j)\right]\label{eq:jensen}\\
&= \mathbb{E}\left[\max_{j=1}^m \exp(sY_j)\right] \\
&\leq \sum_{j=1}^m \mathbb{E}\left[\exp(sY_j)\right]\label{eq:max-sum}\\
&\leq m\left(\frac{1}{1-s}\right)^q,
\end{align}
where \eqref{eq:jensen} is due to Jensen's inequality and \eqref{eq:max-sum} is true since the maximum is upper bounded by the sum. As a result,
\begin{align}
\mathbb{E}\left[\max_{j=1}^m Y_j\right] &\leq \frac{\ln m}{s} + q \cdot \frac{-\ln(1-s)}{s}\\
& \leq \frac{\ln k}{s} + q \cdot \frac{-\ln(1-s)}{s},\label{eq:upper}
\end{align}
where we have used the assumption that $m\le k$.
Since we assume that $q = o(\log k)$, we can write $q$ as $q = \ell(k)\ln k$ where $\ell(k) \to 0^+$ as $k \to \infty$. Let $s = 1 - \ell(k)$ in \eqref{eq:upper}, then
\begin{align}
\mathbb{E}\left[\max_{j=1}^m Y_j\right] &\leq \frac{\ln k}{1 - \ell(k)}\left(1 - \ell(k)\ln\ell(k)\right)\\
&=(\ln k)\left(1+\frac{\ell(k)}{1-\ell(k)}\right)\left(1 - \ell(k)\ln\ell(k)\right).
\end{align}
Note that $\lim_{k \to \infty} \ell(k)\ln\ell(k) = 0$.  Then as $k \to \infty$,
$$
\mathbb{E}\left[\max_{j=1}^m Y_j\right] \leq (\ln k)(1 + o(1)),
$$
which completes the proof. 
\end{proof}

\subsection{Proof of Lemma~\ref{lem:filling} (Filling Probability)}\label{app:lem:filling}
\renewcommand{\thmnamerestated}{Lemma~\ref{lem:filling} (Filling Probability) [Restated]}
\newtheorem*{thmrestatedfilling}{\thmnamerestated}{\bfseries}{\itshape}
\begin{thmrestatedfilling}
Under the assumptions of Theorem~\ref{main_theorem_1}, given that the system is in a state $\bm{s}$ such that
\begin{equation}
\sum_{i=1}^{\ell} s_i \leq \ell\left(1 - \frac{1}{4}\beta N^{-\alpha}\right),
\end{equation}
the probability of the event $\mathrm{FILL}_{\ell}$ for any $\ell\in\{h-1,h,b\}$ can be bounded as $\Pr\left\{\mathrm{FILL}_{\ell}\right\} \geq 1 - \frac{1}{N}$ when $N$ is sufficiently large.
\end{thmrestatedfilling}

\begin{proof}
Assume that a job arrival sees a state $\bm{S}=\bm{s}$ that satisfies 
\begin{equation*}
\sum_{i=1}^{\ell} s_i \leq \ell\left(1 - \frac{1}{4}\beta N^{-\alpha}\right).
\end{equation*}
We focus on the the number of spaces below the threshold $\ell$ in the sampled queues, denoted by $N_{\ell}$. Then $N_{\ell}$ is the maximum number of tasks that can be put into these queues such that all of these tasks are at queueing positions below $\ell$. Therefore, 
$$
\Pr\left\{\mathrm{FILL}_{\ell}\right\} = \Pr\left\{N_{\ell} \geq k\right\} \geq 1 - \Pr\left\{N_{\ell} \leq k\right\}.
$$

\begin{sloppypar}
Now we bound $\Pr\left\{N_{\ell} \leq k\right\}$.  We can think of the sampling process of batch-filling as sampling $kd$ queues one by one without replacement.  Let $X_1,X_2,\cdots,X_{kd}$ be the numbers of spaces below $\ell$ in the $1$st, $2$nd, \dots, $kd$th sampled queues, respectively.  Then $N_{\ell}=X_1+\dots+X_{kd}$.  It is not hard to see that for each of the sampled queue and each integer $x$ with $1 \leq x \leq \ell$,  
\begin{equation*}
\Pr\{X_i=x\}=s_{\ell-x}-s_{\ell-x+1},
\end{equation*}
and $\Pr\{X_i=0\}=s_{\ell}$.
\end{sloppypar}

\begin{sloppypar}
Note that since we sample without replacement, $X_1,X_2,\dots,X_{kd}$ are not independent.  But we can still derive concentration bounds using a result of Hoeffding \cite[Theorem~4]{Hoe_63}.  By this result, we have $\expect\left[f\left(\sum_{i=1}^{kd}X_i\right)\right]\le \expect\left[f\left(\sum_{i=1}^{kd}Y_i\right)\right]$ for any continuous and convex function $f(\cdot)$, where $Y_1,Y_2,\dots,Y_{kd}$ are i.i.d.\ and follow the same distribution as $X_1$. We take the function $f(\cdot)$ to be $f(x)=e^{-tx}$ with $t>0$.  Then
\begin{align*}
&\mspace{23mu}\Pr\left\{N_{\ell} \leq k\right\}\\
&= \Pr\left\{e^{-tN_{\ell}} \geq e^{-tk}\right\} \\
&\leq e^{tk}\prod_{i=1}^{kd} \mathbb{E}\left[e^{-tY_i}\right] \\
&= e^{tk}\prod_{i=1}^{kd}\left(1 - \sum_{j=1}^{\ell}\left(s_{\ell -j}-s_{\ell-j+1}\right)\left(1 - e^{-tj}\right)\right).
\end{align*}
Since $1-x\le e^{-x}$ for each $x\ge 0$, this can be further bounded as
\begin{align}
&\mspace{23mu}\Pr\left\{N_{\ell} \leq k\right\}\nonumber\\
&\le \exp\left(tk - kd\sum_{j=1}^{\ell} \left(s_{\ell-j}-s_{\ell-j+1}\right)\left(1-e^{-tj}\right)\right)\nonumber\\
&\le \exp\left(tk + kd\sum_{j=1}^{\ell} \left(s_{j - 1} - s_j\right)\left(e^{-t(\ell - j + 1)} - 1\right)\right).\label{lemma2:1}
\end{align}
Rearranging the terms in the sum in \eqref{lemma2:1}, we get
\begin{align}
&\mspace{23mu}\sum_{j=1}^{\ell} \left(s_{j - 1} - s_j\right)\left(e^{-t(\ell - j + 1)} - 1\right)\nonumber\\
&= \left(e^{-t\ell} - 1\right) + \left(e^{t}-1\right)\sum_{j=1}^{\ell} s_je^{-t(\ell-j+1)}.\label{lemma2:2}
\end{align}
Since $1 \geq s_1 \geq \cdots s_{\ell}$ and we have assumed that $\sum_{j = 1}^{\ell} s_j \leq \ell\left(1 - \frac{1}{4}\beta N^{-\alpha})\right)$, \eqref{lemma2:2} is maximized when 
$$
s_1 = s_2 = \cdots = s_{\ell} = 1 - \frac{1}{4}\beta N^{-\alpha}. 
$$
Therefore, the upper bound becomes
$$
\Pr\left\{N_{\ell}\le k\right\} \leq \exp\left(tk + kd\left(e^{-t\ell} - 1\right)\frac{1}{4}\beta N^{-\alpha}\right).
$$

Now we apply the condition that $d \geq \frac{8N^{\alpha}}{\beta h}$ and let $t = \frac{\ln(2\ell) - \ln h}{\ell}$.  Then
\begin{align*}
&\mspace{23mu}\Pr\left\{N_{\ell} \leq k\right\}\\
&\le \exp\left(tk + \frac{2k}{h}\left(e^{-t\ell} - 1\right)\right)\\
&= \exp\left(\frac{k}{h}\left(\frac{h}{\ell}\left(\ln(2\ell) - \ln h\right)+\frac{h}{\ell}-2\right)\right).
\end{align*}
Recall the we have assumed that $\frac{k}{h} = \omega(\log N)$ and $h=\omega(1)$.  Then it can be verified that with a sufficiently large $N$, $\frac{h}{\ell}\left(\ln(2\ell) - \ln h\right)+\frac{h}{\ell}+2N^{-0.5}-2$ is smaller than a negative constant for all $\ell\in\{h-1,h,b\}$.  Thus
$$
\Pr\left\{N_{\ell} \leq k\right\} \leq \exp(-\omega(\log N)) \leq \frac{1}{N}.
$$
As a result,
$$
\Pr\left\{\mathrm{FILL}_{\ell}\right\}\ge 1-\Pr\{N_{\ell}\le k\} \geq 1 - \frac{1}{N},
$$
which completes the proof. 
\end{sloppypar}
\end{proof}

\subsection{Proof of Lemma~\ref{lem:ssc}}\label{app:lemma:ssc}
Our proof of Lemma~\ref{lem:ssc} relies on Lemma~\ref{lem:tail} below. 
Lemma~\ref{lem:tail} slightly generalizes the well-known Lyapunov-based tail bounds (see, e.g., \cite{b4}, \cite{b5} and \cite{b6}) in that it allows different drift bounds depending on whether a state $\bm{s}$ is in a set $\mathcal{E}$ or not.  In our proof of Lemma~\ref{lem:ssc}, we only need to let $\mathcal{E}$ be the whole state space.  But this generalization will be needed in the proof of impossibility results in Section~\ref{sec:impossibility}.  We omit the proof of Lemma~\ref{lem:tail} since it only needs minor modification to the arguments used in proving the well-known existing bounds.

\begin{lemma}\label{lem:tail}
Consider a continuous time Markov chain $\left\{\bm{S}(t):t\ge 0\right\}$ with a finite state space $\mathcal{S}$ and a unique stationary distribution $\pi$.  For a Lyapunov function $V: \mathcal{S} \rightarrow [0,+\infty)$, define the drift of $V$ at a state $\bm{s} \in \mathcal{S}$ as 
$$
\Delta V(\bm{s}) = \sum_{\bm{s}' \in \mathcal{S}, \bm{s} \not = \bm{s}'} r_{\bm{s} \to \bm{s}'}(V(\bm{s}') - V(\bm{s})),
$$
where $r_{\bm{s} \to \bm{s}'}$ is the transition rate from state $\bm{s}$ to $\bm{s}'$.  Suppose that
$$
\begin{aligned}
\nu_{\mathrm{max}} &:= \sup_{\bm{s},\bm{s}' \in \mathcal{S}:r_{\bm{s} \to \bm{s}'} > 0} |V(\bm{s}) - V(\bm{s}')| < \infty \\
f_{\mathrm{max}} &:= \max\left\{0,\sup_{\bm{s} \in \mathcal{S}} \sum_{\bm{s}':V(\bm{s}') > V(\bm{s})} r_{\bm{s} \to \bm{s}'}\left(V(\bm{s}') - V(\bm{s})\right)\right\} < \infty.
\end{aligned}
$$
Then if there is a set $\mathcal{E}$ with $B > 0,\gamma > 0, \delta \geq 0$ such that 
\begin{itemize}
\item $\Delta V(\bm{s}) \leq -\gamma$ when $V(\bm{s}) \geq B$ and $\bm{s} \in \mathcal{E}$,
\item $\Delta V(\bm{s}) \leq \delta$ when $V(\bm{s}) \geq B$ and $\bm{s} \not \in \mathcal{E}$,
\end{itemize}
it holds that for all $j \in \mathbb{N}$,
$$
\Pr\left\{V(\bm{s}) \geq B + 2 \nu_{\mathrm{max}}j\right\} \leq \left(\frac{f_{\mathrm{max}}}{f_{\mathrm{max}} + \gamma}\right)^j + \left(\frac{\delta}{\gamma} + 1\right)\Pr\left\{s \not \in \mathcal{E}\right\}.
$$
\end{lemma}

Now we are ready to present to proof of Lemma~\ref{lem:ssc}.
\renewcommand{\thmnamerestated}{Lemma~\ref{lem:ssc} (State-Space Collapse) [Restated]}
\newtheorem*{thmrestatedssc}{\thmnamerestated}{\bfseries}{\itshape}
\begin{thmrestatedssc}
Under the assumption of Theorem \ref{main_theorem_1}, consider the following Lyapunov function:
$$
V(\bm{s}) = \min\left\{\frac{1}{h - 1}\sum_{i=h}^b s_i, b\left(\left(1 - \frac{1}{2}\beta N^{-\alpha}\right) - \frac{1}{h-1}\sum_{i=1}^{h - 1} s_i\right)^+\right\},
$$
where the superscript $^+$ denotes the function $x^+=\max\{x,0\}$.
Let $B = \frac{b - h + 1}{h - 1}\left(\beta N^{-\alpha} + \frac{\log N}{\sqrt{N}}\right)$. Then for any state $\bm{s}$ such that $V(\bm{s}) > B$, its Lyapunov drift can be upper bounded as follows
$$
\Delta V(\bm{s}) = GV(\bm{s}) \leq -\frac{b}{\sqrt{N}}.
$$
Consequently, when $N$ is sufficiently large,
\begin{align*}
\Pr\left\{V(\bm{S}) > B + \frac{2kb\log^2 N}{(h-1)\sqrt{N}}\right\}\le e^{-\frac{1}{2}\log^2 N}.
\end{align*}
\end{thmrestatedssc}

\begin{proof}
Consider the Lyapunov function in the lemma, i.e.,
$$
V(\bm{s}) = \min\left\{\frac{1}{h - 1}\sum_{i=h}^b s_i, b\left(\left(1 - \frac{1}{2}\beta N^{-\alpha}\right) - \frac{1}{h-1}\sum_{i=1}^{h - 1} s_i\right)^+\right\}.
$$

We will refer to the first term and second term in the minimum as $\terml$ and $\termr$, respectively.
Let $B = \frac{b - h + 1}{h - 1}\left(\beta N^{-\alpha} + \frac{\log N}{\sqrt{N}}\right)$ and suppose $V(\bm{s}) > B$.  Recall that the drift of $V$ is given by
\begin{equation*}
\Delta V(\bm{s}) = GV(\bm{s}) = \sum_{\bm{s}' \in \mathcal{S}, \bm{s} \not = \bm{s}'} r_{\bm{s} \to \bm{s}'}(V(\bm{s}') - V(\bm{s})),
\end{equation*}
where $r_{\bm{s} \to \bm{s}'}$ is the transition rate from state $\bm{s}$ to $\bm{s}'$. Let $e_i = \left(0,\cdots,0,\frac{1}{N},0,\cdots,0\right)$ be a vector of length $b$ whose $i$th entry is $\frac{1}{N}$ and all the other entries are zero.  
We divide the discussion into two cases.

\noindent\textbf{Case 1:} $\terml\le \termr$.  In this case $V(\bm{s})=\terml$.  When the state transition is due to a task departure from a queue of length $i$, which has a rate of $N\left(s_{i}-s_{i+1}\right)$, then
\begin{equation*}
V(\bm{s} - e_i) =
\begin{cases}
V(\bm{s}), &\text{ if }1\le i<h,\\
 V(\bm{s}) - \frac{1}{N(h - 1)}, &\text{ if } h\le i\le b.
\end{cases}
\end{equation*}
Now consider the state transition due to a job arrival.  Let $a_i$ be the queueing position that task $i$ is assigned to.  Then the next state can be written as 
$$
\bm{s} + e_{a_1} + \cdots + e_{a_k}.
$$
Note that when the event $\mathrm{FILL}_{h-1}$ happens, the dispatcher puts all $k$ tasks to positions below threshold $h - 1$. Then under $\mathrm{FILL}_{h-1}$, $s_i$ does not change for $i\ge h$, which implies that
$$
V(\bm{s} + e_{a_1} + \cdots + e_{a_k}) = V(\bm{s}).
$$
We can show that $\Pr\left\{\mathrm{FILL}_{h - 1}\right\} \geq 1 - \frac{1}{N}$ using Lemma~\ref{lem:filling} since $\termr\ge \terml > B >0$.  Otherwise, i.e., when $\mathrm{FILL}_{h-1}$ is not true, it is easy to see that
\begin{equation*}
V(\bm{s} + e_{a_1} + \cdots + e_{a_k})\le V(\bm{s})+\frac{k}{N(h-1)}.
\end{equation*}

Therefore,
\begin{align*}
\Delta V(\bm{s}) &\leq \sum_{i=1}^b N(s_{i} - s_{i + 1})\left(V(\bm{s} - e_i) - V(\bm{s})\right)+\frac{N\lambda}{k}\frac{1}{N}\frac{k}{N(h - 1)}\\
&= \frac{1}{N(h-1)}-\frac{s_h}{h - 1}\\
&\le \frac{1}{N(h-1)}-\frac{1}{h - 1}\frac{1}{b-h+1}\sum_{i=h}^bs_i.
\end{align*}
By the assumption that $\terml>B$, we have
\begin{align}
\frac{1}{b-h+1}\sum_{i=h}^bs_i &\ge \frac{h-1}{b-h+1}B=\beta N^{-\alpha}+\frac{\log N}{\sqrt{N}}.\nonumber
\end{align}
Inserting this back to the upper bound on $\Delta V(\bm{s})$ gives
\begin{align*}
\Delta V(\bm{s})&\le -\frac{1}{h-1}\left(-\frac{1}{N}+\beta N^{-\alpha}+\frac{\log N}{\sqrt{N}}\right).
\end{align*}
Since $\frac{\beta N^{-\alpha}}{h-1}\ge \frac{N^{-\alpha}}{k}\ge \frac{b}{\sqrt{N}}$ and $\frac{\log N}{\sqrt{N}}\ge \frac{1}{N}$ when $N$ is sufficiently large, this upper bound becomes
\begin{equation*}
\Delta V(\bm{s})\le -\frac{b}{\sqrt{N}}.
\end{equation*}

\noindent\textbf{Case 2:} $\terml > \termr$. In this case $V(\bm{s})=\termr$. Similarly, a task departs from a queue of length $i$ at a rate of $N(s_i-s_{i+1})$.  The change in $V(\bm{s})$ can be bounded as
$$
V(\bm{s} - e_i)-V(\bm{s})\le
\begin{cases}
\frac{b}{N(h-1)}, & \text{ if }1\le i < h,\\
0, \text{ if } h\le i \le b.
\end{cases}
$$
When a job arrives, under the event $\mathrm{FILL}_{h - 1}$,
\begin{equation*}
V(\bm{s} + e_{a_1} + \cdots + e_{a_k})= V(\bm{s})-\frac{kb}{N(h-1)},
\end{equation*}
where we have used the fact that $\termr > B$.  Again, $\Pr\left\{\mathrm{FILL}_{h - 1}\right\} \geq 1 - \frac{1}{N}$ by Lemma~\ref{lem:filling}.  Otherwise, i.e., when $\mathrm{FILL}_{h - 1}$ is not true, $V(\bm{s} + e_{a_1} + \cdots + e_{a_k})\le V(\bm{s})$.

Therefore,
\begin{align}
\Delta V(\bm{s}) &\leq 
\sum_{i=1}^b N(s_{i} - s_{i + 1})\left(V(\bm{s} - e_i) - V(\bm{s})\right)
+\frac{N\lambda}{k}\left(1-\frac{1}{N}\right)\left(-\frac{kb}{N(h-1)}\right)\nonumber\\
&\leq \frac{b}{h-1}\left(s_1-s_h\right)-\frac{b}{h-1}\left(1-\frac{1}{N}\right)\left(1 - \beta N^{-\alpha}\right)\nonumber\\
&\leq \frac{b}{h-1}\left(1 - \left(\beta N^{-\alpha} + \frac{\log N}{\sqrt{N}}\right)-\left(1 - \frac{1}{N}\right)\left(1 - \beta N^{-\alpha}\right)\right),\label{eq:sbounds}\\
&=\frac{b}{h-1}\left(-\frac{\log N}{\sqrt{N}}+ \frac{1}{N}\left(1-\beta N^{-\alpha}\right)\right)\nonumber\\
&\leq -\frac{b}{h-1}\frac{\log N -\frac{1}{\sqrt{N}}}{\sqrt{N}},\nonumber
\end{align}
where \eqref{eq:sbounds} is due to the fact that $s_1\le 1$ and the fact that $s_h\ge \beta N^{-\alpha} + \frac{\log N}{\sqrt{N}}$ following similar arguments as those in Case 1 noting that $\terml > \termr > B$.  When $N$ is sufficiently large, this upper bound becomes
\begin{equation*}
\Delta V(\bm{s})\le -\frac{b}{\sqrt{N}},
\end{equation*}
which completes the proof of the drift bound in Lemma~\ref{lem:ssc}.

For this Lyapunov function $V$, under the notation in Lemma~\ref{lem:tail}, we have that $\nu_{\mathrm{max}} \leq \frac{kb}{N(h-1)}$ and $f_{\mathrm{max}} \leq \frac{b}{h-1}$.  Let $\mathcal{E} = \mathcal{S}$ and $j = \sqrt{N}\log^2 N$.  Then by Lemma~\ref{lem:tail}, the drift bound implies that
\begin{align*}
&\mspace{21mu}\Pr\left\{V(\bm{S}) > B + \frac{2kb\log^2 N}{(h-1)\sqrt{N}}\right\}\\
&=\Pr\left\{V(\bm{S}) > B + \frac{2kb}{(h-1)N}j\right\}\\
&\leq \left(1 + \frac{h-1}{\sqrt{N}}\right)^{-j} \\
&\leq \left(\left(1 + \frac{1}{\sqrt{N}}\right)^{\sqrt{N} + 1}\right)^{-\frac{1}{\sqrt{N}+1}\sqrt{N}\log^2 N} \\
&\leq e^{-\frac{1}{2}\log^2 N},
\end{align*}
where the last inequality holds when $N$ is sufficiently large.  This completes the proof.
\end{proof}

\section{Lemmas needed for impossibility results}\label{app:lemmas-impossibility}

\subsection{Lemma~\ref{lower:lemma0}}
\begin{lemma}\label{lower:lemma0}
Assume that the system is stable. Then for any $x > 0$, 
$$
\Pr\left\{S_1 < 1 - x\right\} \leq \frac{\beta N^{-\alpha}}{x}.
$$
\end{lemma}
\begin{proof}
By work conservation law, it holds that
$
\mathbb{E}[S_1] = \lambda = 1 - \beta N^{-\alpha}.
$
Then $\mathbb{E}[1-S_1]=\beta N^{-\alpha}.$
Therefore, by the Markov inequality, for any $x>0$,
$$
\Pr\left\{S_1 < 1-x\right\}=\Pr\left\{1-S_1 > x\right\}\leq \frac{\beta N^{-\alpha}}{x}.
$$
\end{proof}

\subsection{Lemma~\ref{lower:lemma1}}
\begin{lemma}\label{lower:lemma1}
Let $\ell$ be a threshold such that $1 \leq \ell \leq h$ with $h=O(\log k)$. Suppose that an incoming job sees a state 
$\bm{s}$ such that $\sum_{i=1}^{\ell} s_i \geq \ell - x$,
where $x = \Omega(hN^{-\alpha})$ and $x = e^{-\Omega(\log N)}$. Consider a Lyapunov function 
$
V_{\ell}(\bm{s}) = s_1 + s_2 + \cdots + s_{\ell}.
$
It holds that when $N$ is sufficiently large,
$$
\sum_{\bm{s}':\bm{s} \to \bm{s}' \text{~due to an arrival}} r_{\bm{s} \to \bm{s}'}\left(V_{\ell}(\bm{s}') - V_{\ell}(\bm{s})\right) \leq 2kdx,
$$
where $r_{\bm{s} \to \bm{s}'}$ is the transition rate, and $\bm{s} \to \bm{s}' \text{~due to an arrival}$ means that $\bm{s}$ will move to state $\bm{s}'$ on the Markov chain only if there is an incoming job.
\end{lemma}
\begin{proof}
Suppose that an arrival sees a state $\bm{s}$. Given $\sum_{i=1}^{\ell} s_i \geq \ell - x$, we have $s_{\ell} \geq 1 - x$ since $s_i \leq 1$ for all $1 \leq i \leq \ell$. Without loss of generality, we can think of the batch-filling policy as sampling the $kd$ queues one by one.  During the sampling, we always choose at most $kd$ servers of length at least $\ell$. The probability that all $kd$ sampled servers have length at least $\ell$ is thus larger or equal to
$$
\left(\frac{N(1-x)-kd}{N}\right)^{kd} = \left(1 - \left(x + \frac{kd}{N}\right)\right)^{kd}.
$$
Recall that by the assumptions in Theorem~\ref{main_theorem_2}, we have $x = e^{-\Omega(\log N)}, kd = o(N^{1 - \alpha})$, and thus $x + \frac{kd}{N} > -1$ when $N$ is sufficiently large. Furthermore, applying Bernoulli's Inequality and the assumption that $x = \Omega(hN^{-\alpha})$, it holds
$$
\left(1 - \left(x + \frac{kd}{N}\right)\right)^{kd} \geq 1 - kd\left(x + \frac{kd}{N}\right) \geq 1-2xkd
$$
for a large $N$. 
Note that if we put all tasks of this arrival into servers of length at least $\ell$, we will not affect the value of $V_l(\bm{s})$. As a result,
$$
\begin{aligned}
&\sum_{\bm{s}':\bm{s} \to \bm{s}' \text{~due to an arrival}} r_{\bm{s} \to \bm{s}'}\left(V_{\ell}(\bm{s}') - V_l(\bm{s})\right) \\
\leq& \left(1-2kdx\right)\cdot 0 \cdot \frac{\lambda}{k} + 2kdx \cdot k\frac{\lambda}{k} \\
\leq& 2kdx,
\end{aligned}
$$
which completes the proof. 
\end{proof}

\subsection{Lemma~\ref{lower:lemma2}}
Lemma~\ref{lower:lemma2} is a key in establishing the inductive proof.  This lemma relates $S_q$ to $S_{q-1}$ for $3 \leq i \leq h$.

\begin{lemma}\label{lower:lemma2}
Define $u = 2kd$ and $b_q = u^{q-1}h^q\beta N^{-\alpha}$ for $q \in \mathbb{N}$. Define a sequence $a_q$, such that $a_1 = 0,a_2 = 1$ and $a_q = (q-2)a_{q-1} + 2$ for $q > 2$.  For any $q$ with $3 \leq q \leq h$, if
$$
\Pr\left\{S_1 - S_{q - 1} \leq a_{q-1}b_{q-1}\right\} \geq \left(\frac{h-2}{h}\right)^{q-2}-(q-2)N^{-\log N},
$$
then
$$
\Pr\left\{S_1 - S_{q} \leq a_qb_q\right\} \geq \left(\frac{h-2}{h}\right)^{q-1}-(q-1)N^{-\log N}.
$$
\end{lemma}
\begin{proof}
The proof is close to that of Theorem \ref{main_theorem_2}. Recall that for each $1 \leq \ell \leq h$ and state $\bm{s} \in \mathcal{S}$, we define the Lyapunov function
$$
V_{\ell}(\bm{s}) = \sum_{i=1}^{\ell} s_i.
$$
For $q$ such that $3 \leq q \leq h$, by assumption, 
$$
\Pr\left\{S_1 - S_{q - 1} \leq a_{q-1}b_{q-1}\right\} \geq \left(\frac{h-2}{h}\right)^{q-2} - (q-2)N^{-\log N}.
$$
It holds
\begin{equation}\label{lower:lemma2:eq1}
\begin{aligned}
&\mspace{21mu}\Pr\left\{V_{q-1}(\bm{S}) < q - 1 - \left((q-2)a_{q-1}+1\right)b_{q-1}\right\} \\ 
&\leq \Pr\left\{V_{q-1}(\bm{S}) < q - 1 - \left((q-2)a_{q-1}+1\right)b_{q-1}, \right. \\
&\mspace{21mu}\left. S_1 - S_{q - 1} \leq a_{q-1}b_{q-1}\right\} \\
&\mspace{21mu}+\Pr\left\{S_1 - S_{q- 1} > a_{q - 1}b_{q-1}\right\} \\
&\leq \Pr\left\{(q-1)S_1 < q-1-b_{q-1}\right\} + 1 - \left(\frac{h-2}{h}\right)^{q-2}\\
&\mspace{21mu}+(q-2)N^{-\log N} \\
&\leq \frac{q-1}{u^{q-2}h^{q-1}}+ 1 - \left(\frac{h-2}{h}\right)^{q-2}+(q-2)N^{-\log N}.
\end{aligned}
\end{equation}
The last inequality uses Lemma \ref{lower:lemma0} and $b_{q - 1} = u^{q-2}h^{q-1}\beta N^{-\alpha}.$

Now let $B_{q - 1} = q - 1 - \left((q - 2)a_{q - 1} + 2\right)b_{q - 1}$. We can see that $B_{q - 1} = q - 1 - a_qb_{q-1}$. For a state $\bm{s}$ such that $V_{q-1}(\bm{s}) > B_{q-1}$, it holds
$$
\begin{aligned}
\Delta V_{q-1}(\bm{s}) &=  \sum_{\bm{s}':\bm{s} \to \bm{s}' \text{~due to an arrival}} r_{\bm{s} \to \bm{s}'}\left(V_{q-1}(\bm{s}') - V_{q-1}(\bm{s})\right)\\ &\mspace{21mu}+\sum_{\bm{s}':\bm{s} \to \bm{s}' \text{~due to a departure}} r_{\bm{s} \to \bm{s}'}\left(V_{q-1}(\bm{s}') - V_{q-1}(\bm{s})\right).
\end{aligned}
$$
Recall that we define $u = 2kd$ and $b_{q} = u^{q-1}h^{q}\beta N^{-\alpha}$. As $V_{q-1}(\bm{s}) > q - 1 - a_qb_{q - 1}$, by Lemma \ref{lower:lemma1}, it holds
$$
\begin{aligned}
\Delta V_{q-1}(\bm{s}) &\leq 2kda_qb_{q - 1} - (s_1 - s_q) \\
&=a_qu^{q-1}h^{q-1}\beta N^{-\alpha} - (s_1-s_q).
\end{aligned}
$$
Let $\Pr\left\{S_1-S_q \leq a_qb_q\right\}=p_q, \mathcal{E}_{q-1} = \left\{s \in \mathcal{S} \mid s_1 - s_q > a_qb_q\right\}$. Then $\Pr\left\{S \not \in \mathcal{E}_{q - 1}\right\} = p_q$. For a state $\bm{s}$, consider the following two cases.
\begin{itemize}
\item $\bm{s} \not \in \mathcal{E}_{q - 1}$, $\Delta V_{q-1}(\bm{s}) \leq a_qu^{q-1}h^{q-1}\beta N^{-\alpha}\eqqcolon\delta$.
\item $\bm{s} \in \mathcal{E}_{q - 1}$. Let $\gamma = -\Delta V_{q-1}(\bm{s})$. It holds $$\gamma \geq a_qu^{q-1}h^{q-1}\beta N^{-\alpha}(h-1).$$
\end{itemize}
We then utilize the tail bound, Lemma \ref{lem:tail}. Following the definition in Lemma \ref{lem:tail}, it is easy to verify that $\nu_{\mathrm{max}} \leq \frac{k}{N}, f_{\mathrm{max}} \leq 1$ for the Lyapunov function $V_{q-1}(\bm{s})$. Let 
$$j_{q - 1} = \left(1+\frac{N^{\alpha}}{a_qu^{q-1}h^{q-1}(h-1)\beta}\right)\log^2 N.$$
Using Lemma \ref{lem:tail},
$$
\begin{aligned}
&\mspace{21mu}\Pr\left\{V_{q-1}(\bm{S})>B_{q-1}+2\nu_{\mathrm{max}}j_{q-1}\right\} \\
&\leq  \left(\frac{f_{\mathrm{max}}}{f_{\mathrm{max}} + \gamma}\right)^{j_{q-1}}+\left(\frac{\delta}{\gamma}+1\right)\Pr\left\{\bm{S} \not \in \mathcal{E}_{q-1}\right\} \\
&\leq \left(\frac{f_{\mathrm{max}}}{f_{\mathrm{max}} + \gamma}\right)^{j_{q-1}} + \frac{h}{h-1}p_{q}.
\end{aligned}
$$
Note that when $N$ is sufficiently large, 
$$
\left(\frac{f_{\mathrm{max}}}{f_{\mathrm{max}} + \gamma}\right)^{j_{q-1}} \leq e^{-\log^2 N}.
$$
Besides, we assume that $0 < \alpha < 0.5, k = e^{o(\sqrt{\log N})}$ and $h = O(\log k)$. As a result, for a large $N$,
$$
\begin{aligned}
&\Pr\left\{V_{q-1}(\bm{S}) \geq q - 1 - ((q-2)a_{q-1}+1)b_{q-1}\right\} \\
&\leq \Pr\left\{V_{q-1}(\bm{S})>B+2\nu_{\mathrm{max}}j_{q-1}\right\} \\
&\leq e^{-\log^2 N} + \frac{h}{h-1}p_q.\\
\end{aligned}
$$
Together with Eq.(\ref{lower:lemma2:eq1}), we have
$$
\begin{aligned}
&\left(\frac{h-2}{h}\right)^{q-2}-\frac{q-1}{u^{q-2}h^{q-1}}-(q-2)N^{-\log N} \\
&\leq \Pr\left\{V_{q-1}(\bm{S}) > q - 1 - ((q-2)a_{q-1}+1)b_{q-1}\right\} \\
&\leq e^{-\log^2 N} + \frac{h}{h-1}p_q.
\end{aligned}
$$
We can conclude that for a large $N$,
$$
\Pr\left\{S_1 - S_q \leq a_qb_q\right\}=p_q \geq \left(\frac{h-2}{h}\right)^{q-1}-(q-1)N^{-\log N},
$$
which completes the proof.
\end{proof}

\subsection{Lemma~\ref{lower:lemma3}}
Lemma~\ref{lower:lemma3} complements the probability bound in Lemma~\ref{lem:filling}. Recall that $\mathrm{FILL}_h$ denotes the event that all the $k$ tasks of an incoming job are assigned to queueing positions below a threshold $h$.  Lemma~\ref{lower:lemma3} gives a condition on the total queue length for $\mathrm{FILL}_h$ to happen with low probability.

\begin{lemma}\label{lower:lemma3}
Suppose an incoming job sees a state $\bm{s}$ such that $\sum_{i=1}^h s_i > h - \frac{1}{3d}$. Then when $N$ is sufficiently large,
$$
\Pr\left\{\mathrm{FILL}_h\right\} = o(1).
$$
\end{lemma}
\begin{proof}
We use a similar argument as the proof of Lemma \ref{lem:filling}. Suppose that an arrival sees a state $\bm{s}$. By assumption, it holds
$$
\sum_{i=1}^h s_i \geq h - \frac{1}{3d}.
$$
Let $X_1,\cdots,X_{kd}$ be the numbers of places below $h$ in each sampled server. The goal is to show
$$
\Pr\left\{\mathrm{FILL}_h\right\} = \Pr\left\{\sum_{i=1}^{kd} X_i \geq k\right\} = o(1)
$$
when $N$ is large enough.

We could see that for each integer $x$ such that $1 \leq x \leq h$, $\Pr\{X_i = x\} = s_{h - x} - s_{h - x + 1}$, and $\Pr\{X_i = 0\} = s_h$. Since we are sampling without replacement, $X_1,\cdots,X_{kd}$ are not independent. But still, utilizing a result of Hoeffding \cite[Theorem~4]{Hoe_63}, we have $\expect\left[f\left(\sum_{i=1}^{kd} X_i\right)\right] \leq \expect\left[f\left(\sum_{i=1}^{kd} Y_i\right)\right]$ for any continuous and convex function $f(\cdot)$, where $Y_1,\cdots,Y_{kd}$ are i.i.d.\ and follow the same distribution as $X_1$. Take $f(\cdot)$ to be $f(x) = e^{tx}$ where $t$ is some positive value.   

It then holds
$$
\begin{aligned}
\Pr\left\{\mathrm{FILL}_h\right\} &=\Pr\left\{\sum_{i=1}^{kd} X_i \geq k\right\} \\
&= \Pr\left\{e^{t\sum_{i=1}^{kd} X_i} \geq e^{tk}\right\} \\
&\leq e^{-tk}\prod_{i=1}^{kd} \expect\left[e^{tY_i}\right] \\
&=e^{-tk}\prod_{i=1}^{kd} \left(1 + \sum_{j=1}^h \left(e^{t(h-j+1)-1}-1\right)\right).
\end{aligned}
$$
Since for all $x > 0$, $1 + x \leq e^x$, we can further have
\begin{equation}\label{lower:lemma3:eq1}
\Pr\left\{\mathrm{FILL}_h\right\}
\leq e^{-tk}\exp\left(kd\sum_{j=1}^{h} \left(e^{t(h-j+1)}-1\right)(s_{j-1}-s_j)\right).
\end{equation}
Rearranging the sum in (\ref{lower:lemma3:eq1}), we get
\begin{equation}\label{lower:lemma3:eq2}
\begin{aligned}
&\mspace{21mu}\sum_{j=1}^h \left(e^{t(h-j+1)}-1\right)(s_{j-1}-s_j) \\
&=e^{th}-\sum_{j=1}^h s_j\left(e^{t(h-j+1)}-e^{t(h-j)}\right) \\
&=e^{th}-(e^t-1)\sum_{j=1}^h s_je^{t(h-j)}.
\end{aligned}
\end{equation}
Recall that $\sum_{j=1}^h s_j \geq h - \frac{1}{3d}$, and $1 \geq s_1 \geq s_2 \geq \cdots \geq s_h \geq 0$. Eq. (\ref{lower:lemma3:eq2}) is maximized when $s_1 = s_2 = \cdots = s_h = 1 - \frac{1}{3dh}$ and thus,
$$
(\ref{lower:lemma3:eq2}) \leq (e^{th}-1)\frac{1}{3dh}.
$$
Plug it into Inequality (\ref{lower:lemma3:eq1}), 
$$
\Pr\left\{\mathrm{FILL}_h\right\} \leq \min_{t > 0} \exp\left(k\left(-t + \frac{e^{th}-1}{3h}\right)\right).
$$
Pick $t = \frac{\ln 3}{h}$. It holds
$$
\Pr\left\{\mathrm{FILL}_h\right\} 
\leq \exp\left(\frac{k}{3h}\left(-3\ln 3 + 2\right)\right).
$$
By the assumption that $\frac{k}{h} = \omega(1)$, we could conclude that
$$
\Pr\left\{\mathrm{FILL}_h\right\} = o(1)
$$
when $N$ is sufficiently large.
\end{proof}

\section{Proof of Theorem~\ref{thm:lower-zero-waiting}} \label{sec:proof-thm-zero-waiting}
\begin{proof}
Let $\mathcal{I}$ be the event that all the tasks of an incoming job are assigned to idle servers in steady state.  Then what we need to show is $\Pr\{\mathcal{I}\} \leq 0.5$.

From the stability of batch-filling \cite{b3} and the Little's law, it holds $\expect{S_1} = \lambda$. For a job arrival of $k$ tasks, in order to schedule every task to an idle server, batch-filling needs to find at least $k$ idle servers. Suppose batch-filling probes $kd$ servers with state $X_{1}, \cdots, X_{kd}$ where $X_i$ is a $0-1$ random variables indicating whether the sampled $i$th server is idle. Then 
\[
\Pr\{\mathcal{I}\} = \Pr\{X_{1}+\cdots+X_{kd} \geq k\}.
\]
Notice that $\expect\left[X_1+\cdots+X_{kd}\right] = kd(1-\lambda)$ by the linearity of expectations. If $d \leq \frac{1}{2(1-\lambda)}$, this expectation is upper bounded by $\frac{k}{2}$. Therefore,
\[
\Pr\{\mathcal{I}\} = \Pr\{X_{1}+\cdots+X_{kd} \geq k\} \leq \frac{\expect\left[X_1+\cdots+X_{kd}\right]}{k} \leq 0.5.
\]
\end{proof}

\section{More Details on Simulations}
\subsection{Probe Ratios} \label{sec:experiment}
In the simulations, we need to adjust the definition of probe ratio a little. Let $D_i = \lfloor \min(N,kd_i) \rfloor$ for $1 \leq i \leq 4$. Then $D_i$ is the true number of probes used in batch-filling for each job. When $N$ is small, $D_i$ may be equal to $N$. In this case, we adjust the value of $d_i$ as $\frac{D_i}{k}$, which is the true expected probe ratio of each task. The exact value of $d_i$ is shown in Table \ref{table:probe-ratio}.
\begin{table}[!hbp]
\begin{tabular}{|l|l|l|l|l|}
\hline
$N$   & $d_1$ & $d_2$ & $d_3$ & $d_4$ \\ \hline
32    & 2.7   & 2.1   & 2.7   & 2.7   \\ \hline
64    & 3.8   & 2.2   & 3.8   & 3.8   \\ \hline
128   & 5.6   & 2.2   & 4.7   & 5.6   \\ \hline
256   & 7.6   & 2.3   & 5.1   & 8.2   \\ \hline
512   & 9.3   & 2.4   & 5.5   & 9.4   \\ \hline
1024  & 11.3  & 2.5   & 5.9   & 10.8  \\ \hline
2048  & 13.7  & 2.6   & 6.3   & 12.4  \\ \hline
4096  & 16.6  & 2.7   & 6.8   & 14.4  \\ \hline
8192  & 20.2  & 2.8   & 7.3   & 16.7  \\ \hline
16384 & 24.5  & 3.0   & 7.9   & 19.4  \\ \hline
32768 & 29.9  & 3.1   & 8.5   & 22.6  \\ \hline
65536 & 36.5  & 3.2   & 9.3   & 26.5  \\ \hline
\end{tabular}
\caption{Probe Ratios for Different Scales of System}
\label{table:probe-ratio}
\end{table}

\subsection{Numerical Values for Figures~\ref{fig:different-d} and \ref{fig:different-dist}} \label{sec:exactval}
We give the numerical values and standard deviations for Figures~\ref{fig:different-d} and \ref{fig:different-dist} in Tables~\ref{table:diff-d} and \ref{table:diff-dist}, respectively.
\begin{table}[!hbp]
    \centering
\begin{tabular}{|l|l|l|l|l|}
\hline
$N$ & $d_1$                          & $d_2$                         & $d_3$                         & $d_4$                         \\ \hline
$32$    & $0.23(\pm 3.8\times 10^{-4})$  & $0.27(\pm 3.2\times 10^{-4})$ & $0.23(\pm 3.8\times 10^{-4})$ & $0.23(\pm 3.8\times 10^{-4})$ \\ \hline
$64$    & $0.20(\pm 4.5\times 10^{-4})$  & $0.30(\pm 3.0\times 10^{-4})$ & $0.20(\pm 4.5\times 10^{-4})$ & $0.20(\pm 4.5\times 10^{-4})$ \\ \hline
$128$   & $0.17(\pm 2.4 \times 10^{-4})$ & $0.33(\pm 9.9\times 10^{-4})$ & $0.18(\pm 3.6\times 10^{-4})$ & $0.17(\pm 2.4\times 10^{-4})$ \\ \hline
$256$   & $0.14(\pm 3.7\times 10^{-4})$  & $0.35(\pm 6.6\times 10^{-4})$ & $0.19(\pm 1.0\times 10^{-4})$ & $0.14(\pm 1.7\times 10^{-4})$ \\ \hline
$512$   & $0.13(\pm 6.2\times 10^{-4})$  & $0.37(\pm 7.8\times 10^{-4})$ & $0.19(\pm 3.6\times 10^{-4})$ & $0.13(\pm 1.1\times 10^{-4})$ \\ \hline
$1024$  & $0.12(\pm 1.5\times 10^{-4})$  & $0.40(\pm 1.8\times 10^{-4})$ & $0.20(\pm 3.8\times 10^{-4})$ & $0.12(\pm 2.3\times 10^{-5})$ \\ \hline
$2048$  & $0.10(\pm 2.1\times 10^{-4})$  & $0.41(\pm 3.0\times 10^{-4})$  & $0.21(\pm 4.3\times 10^{-4})$ & $0.11(\pm 1.1\times 10^{-4})$ \\ \hline
$4096$  & $0.09(\pm 6.4\times 10^{-4})$  & $0.43(\pm 1.6\times 10^{-3})$ & $0.24(\pm 4.8\times 10^{-4})$ & $0.10(\pm 2.9\times 10^{-4})$ \\ \hline
$8192$  & $0.08(\pm 2.0\times 10^{-4})$  & $0.45(\pm 8.3\times 10^{-4})$ & $0.26(\pm 1.9\times 10^{-4})$ & $0.10(\pm 4.7\times 10^{-4})$ \\ \hline
$16384$ & $0.07(\pm 2.7\times 10^{-4})$  & $0.48(\pm 6.7\times 10^{-4})$ & $0.28(\pm 5.8\times 10^{-4})$ & $0.09(\pm 4.1\times 10^{-4})$ \\ \hline
$32768$ & $0.05(\pm 2.5\times 10^{-4})$  & $0.51(\pm 1.1\times 10^{-3})$ & $0.30(\pm 2.1\times 10^{-4})$ & $0.08(\pm 2.8\times 10^{-4})$ \\ \hline
$65536$ & $0.05(\pm 2.2\times 10^{-4})$  & $0.54(\pm 5.3\times 10^{-4})$ & $0.31(\pm 2.2\times 10^{-4})$ & $0.10(\pm 4.7\times 10^{-4})$ \\ \hline
\end{tabular}
    \caption{Values of $\frac{\expect[T-T^*]}{\expect[T^*]}$ in Figure \ref{fig:different-d}}
    \label{table:diff-d}
\end{table}

\begin{table}[!hbp]
    \centering
\begin{tabular}{|l|l|l|l|l|}
\hline
$N$ & Exponential & Hyper-Exponential  & Bounded Pareto & S \& X \\ \hline
$32$    & $0.23(\pm 3.8\times 10^{-4})$  & $0.84(\pm 2.7\times 10^{-3})$ & $0.49(\pm 1.4\times 10^{-3})$ & $0.55(\pm 8.1\times 10^{-4})$ \\ \hline
$64$    & $0.20(\pm 4.5\times 10^{-4})$  & $0.73(\pm 1.8\times 10^{-3})$ & $0.46(\pm 1.6\times 10^{-3})$ & $0.53(\pm 1.2\times 10^{-3})$ \\ \hline
$128$   & $0.17(\pm 2.4 \times 10^{-4})$ & $0.61(\pm 8.3\times 10^{-4})$ & $0.41(\pm 6.0\times 10^{-4})$ & $0.48(\pm 8.5\times 10^{-4})$ \\ \hline
$256$   & $0.14(\pm 3.7\times 10^{-4})$  & $0.50(\pm 9.5\times 10^{-4})$ & $0.36(\pm 1.2\times 10^{-3})$ & $0.44(\pm 1.8\times 10^{-3})$ \\ \hline
$512$   & $0.13(\pm 6.2\times 10^{-4})$  & $0.42(\pm 4.1\times 10^{-4})$ & $0.31(\pm 2.7\times 10^{-4})$ & $0.40(\pm 4.3\times 10^{-4})$ \\ \hline
$1024$  & $0.12(\pm 1.5\times 10^{-4})$  & $0.34(\pm 6.3\times 10^{-4})$ & $0.27(\pm 2.7\times 10^{-4})$ & $0.34(\pm 5.3\times 10^{-4})$ \\ \hline
$2048$  & $0.10(\pm 2.1\times 10^{-4})$  & $0.27(\pm 4.3\times 10^{-4})$  & $0.23(\pm 4.0\times 10^{-4})$ & $0.28(\pm 9.3\times 10^{-4})$ \\ \hline
$4096$  & $0.09(\pm 6.4\times 10^{-4})$  & $0.19(\pm 9.2\times 10^{-3})$ & $0.19(\pm 2.7\times 10^{-4})$ & $0.23(\pm 8.6\times 10^{-4})$ \\ \hline
$8192$  & $0.08(\pm 2.0\times 10^{-4})$  & $0.12(\pm 4.7\times 10^{-4})$ & $0.13(\pm 4.6\times 10^{-4})$ & $0.18(\pm 7.1\times 10^{-4})$ \\ \hline
$16384$ & $0.07(\pm 2.7\times 10^{-4})$  & $0.07(\pm 9.0\times 10^{-4})$ & $0.10(\pm 4.0\times 10^{-4})$ & $0.12(\pm 5.2\times 10^{-4})$ \\ \hline
$32768$ & $0.05(\pm 2.5\times 10^{-4})$  & $0.03(\pm 4.0\times 10^{-3})$ & $0.03(\pm 2.0\times 10^{-4})$ & $0.08(\pm 3.0\times 10^{-4})$ \\ \hline
$65536$ & $0.05(\pm 2.2\times 10^{-4})$  & $0.01(\pm 2.3\times 10^{-4})$ & $0.02(\pm 1.4\times 10^{-4})$ & $0.06(\pm 1.0\times 10^{-3})$ \\ \hline
\end{tabular}
    \caption{Values of $\frac{\expect[T-T^*]}{\expect[T^*]}$ in Figure \ref{fig:different-dist}}
    \label{table:diff-dist}
\end{table}

\subsection{Delay Scaling when $k = \lfloor \sqrt{N}\rfloor$} \label{sec:newk}
{
In this section, we provide simulation results for the setting where $k = \lfloor \sqrt{N}\rfloor$.  The scalings of probe ratios are the same as in Section~\ref{subsec:scaling}. The results are demonstrated in Figure \ref{fig:newk}, and the numerical values and standard deviations are given in Table \ref{table:newk}.
\begin{figure}[!hbp]
\centering
\scalebox{0.5}{\input{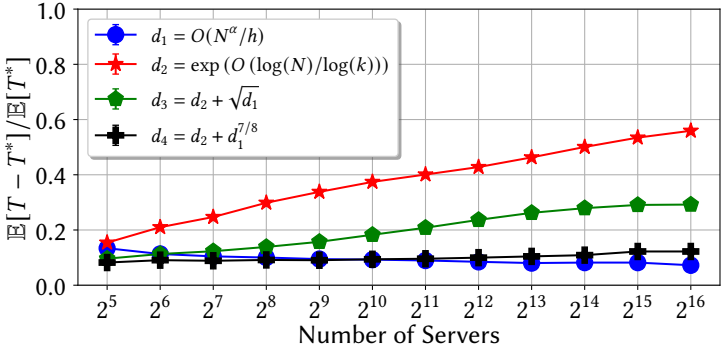}}
\vspace{-0.2in}
\caption{Queueing delays when $k =\lfloor \sqrt{N}\rfloor$ under different probe ratios: $d_1$ is sufficient for convergence to zero queueing delay; $d_1>d_4>d_3>d_2$.}
\label{fig:newk}
\end{figure}
}
\begin{table}[!hbp]
    \centering
\begin{tabular}{|l|l|l|l|l|}
\hline
$N$ & $d_1$                          & $d_2$                         & $d_3$                         & $d_4$                         \\ \hline
$32$    & $0.13(\pm 5.2\times 10^{-4})$  & $0.15(\pm 5.2\times 10^{-4})$ & $0.10(\pm 2.0\times 10^{-4})$ & $0.08(\pm 2.9\times 10^{-4})$ \\ \hline
$64$    & $0.11(\pm 4.2\times 10^{-4})$  & $0.21(\pm 8.1\times 10^{-4})$ & $0.11(\pm 4.1\times 10^{-4})$ & $0.09(\pm 4.1\times 10^{-4})$ \\ \hline
$128$   & $0.10(\pm 7.6 \times 10^{-4})$ & $0.25(\pm 6.5\times 10^{-4})$ & $0.12(\pm 4.8\times 10^{-4})$ & $0.09(\pm 7.2\times 10^{-4})$ \\ \hline
$256$   & $0.10(\pm 3.4\times 10^{-4})$  & $0.30(\pm 4.7\times 10^{-4})$ & $0.14(\pm 7.5\times 10^{-4})$ & $0.09(\pm 3.0\times 10^{-4})$ \\ \hline
$512$   & $0.09(\pm 3.8\times 10^{-4})$  & $0.34(\pm 2.6\times 10^{-4})$ & $0.16(\pm 2.7\times 10^{-4})$ & $0.09(\pm 2.7\times 10^{-4})$ \\ \hline
$1024$  & $0.09(\pm 1.4\times 10^{-4})$  & $0.37(\pm 8.2\times 10^{-4})$ & $0.18(\pm 4.9\times 10^{-4})$ & $0.09(\pm 9.5\times 10^{-5})$ \\ \hline
$2048$  & $0.09(\pm 2.9\times 10^{-4})$  & $0.4(\pm 2.1\times 10^{-4})$  & $0.21(\pm 4.0\times 10^{-4})$ & $0.10(\pm 1.9\times 10^{-4})$ \\ \hline
$4096$  & $0.08(\pm 2.6\times 10^{-4})$  & $0.43(\pm 9.8\times 10^{-4})$ & $0.24(\pm 5.0\times 10^{-4})$ & $0.10(\pm 2.2\times 10^{-4})$ \\ \hline
$8192$  & $0.08(\pm 3.5\times 10^{-4})$  & $0.46(\pm 2.2\times 10^{-4})$ & $0.26(\pm 1.3\times 10^{-4})$ & $0.10(\pm 3.9\times 10^{-4})$ \\ \hline
$16384$ & $0.08(\pm 5.3\times 10^{-4})$  & $0.50(\pm 7.0\times 10^{-4})$ & $0.28(\pm 1.7\times 10^{-4})$ & $0.11(\pm 3.1\times 10^{-4})$ \\ \hline
$32768$ & $0.08(\pm 4.9\times 10^{-4})$  & $0.53(\pm 1.7\times 10^{-4})$ & $0.29(\pm 1.5\times 10^{-4})$ & $0.12(\pm 1.9\times 10^{-4})$ \\ \hline
$65536$ & $0.07(\pm 3.1\times 10^{-4})$  & $0.56(\pm 9.8\times 10^{-4})$ & $0.29(\pm 1.6\times 10^{-4})$ & $0.12(\pm 3.8\times 10^{-4})$ \\ \hline
\end{tabular}
    \caption{Values of  $\frac{\expect[T-T^*]}{\expect[T^*]}$ in Figure \ref{fig:newk}}
    \label{table:newk}
\end{table}


\end{document}